\documentclass[12pt]{article}
\usepackage{amsmath,amsthm}
\usepackage{enumerate}
\usepackage{natbib}
\usepackage{setspace}
\renewcommand{\baselinestretch}{1.5} 
\usepackage{amsmath}
\usepackage{enumerate}
\usepackage{bbm}
\usepackage{pgf,tikz}
\usetikzlibrary{decorations.pathreplacing,angles,quotes}
\usepackage{mathrsfs}

\newtheorem{theorem}{Theorem}

\usepackage{color} 
\usepackage{enumitem} 
\usepackage{mathtools}

\usepackage{titlesec}
\titleformat{\section}{\large\bfseries}{\thesection}{1em}{}
\titleformat{\subsection}{\normalsize\bfseries}{\thesubsection}{1em}{}
\titleformat{\subsubsection}{\normalsize\bfseries}{\thesubsubsection}{1em}{}

\usepackage{pifont}
\usepackage{amssymb}
\usepackage{booktabs}
\usepackage{graphicx}
\usepackage{tabularx}
\usepackage{hyperref}
\usepackage{titlesec}

\titleformat{\paragraph}
{\normalfont\normalsize\bOseries}{\theparagraph}{1em}{}
\titlespacing*{\paragraph}
{0pt}{3.25ex plus 1ex minus .2ex}{1.5ex plus .2ex}

\usepackage{apptools}
\AtAppendix{\counterwithin{lemma1}{section}}

\usepackage[]{appendix}

\allowdisplaybreaks

\usepackage{hyperref}
\usepackage{color}
\usepackage{booktabs} 
\usepackage{xr}
\usepackage{amsmath}
\usepackage{amssymb}
\usepackage{caption}
\usepackage{subcaption}
\usepackage{booktabs}
\usepackage{graphicx}
\usepackage{centernot}
\usepackage{titlesec}
\usepackage{dsfont}
\usepackage{soul}
\usepackage{thmtools, thm-restate}

\usepackage{multirow}
\usepackage{caption}

\usepackage{bbm}
\usepackage{pdflscape, array, rotating, multirow}

\usepackage{centernot}

\usepackage{lscape}
\DeclareMathOperator{\E}{\textnormal{\mbox{E}}}
\def\bSig\mathbf{\Sigma}

\makeatletter
\newcommand*{\indep}{%
  \mathbin{%
    \mathpalette{\@indep}{}%
  }%
}
\newcommand*{\nindep}{%
  \mathbin{
    \mathpalette{\@indep}{\not}
  }%
}
\newcommand*{\@indep}[2]{%
  \sbox0{$#1\perp\m@th$}
  \sbox2{$#1=$}
  \sbox4{$#1\vcenter{}$}
  \rlap{\copy0}
  \dimen@=\dimexpr\ht2-\ht4-.2pt\relax
  \kern\dimen@
  {#2}%
  \kern\dimen@
  \copy0 
} 
\makeatother

\makeatletter

\begin{document}

\def\spacingset#1{\renewcommand{\baselinestretch}%
{#1}\small\normalsize} \spacingset{1}



\begin{center}
\Large{\textbf{Efficient estimation of subgroup treatment effects using multi-source data}}\\
\vspace{1cm}
\small{Guanbo Wang$^ {1, 2,*}$
Alexander W. Levis$^{3}$, Jon A. Steingrimsson$^{4}$, and Issa J. Dahabreh$^ {1,2,5}$ }\\
\vspace{0.5cm}
\footnotesize{$^{1}$CAUSALab, Harvard T.H. Chan School of Public Health, Boston, MA }\\
\footnotesize{$^{2}$Department of Epidemiology, Harvard T.H. Chan School of Public Health, Boston, MA }\\
\footnotesize{$^{3}$Department of Statistics \& Data Science, Carnegie Mellon University, Pittsburgh, PA}\\
\footnotesize{$^{4}$Department of Biostatistics, Brown University, Providence, RI}\\
\footnotesize{$^{5}$Department of Biostatistics, Harvard T.H. Chan School of Public Health, Boston, MA  \\
}
\end{center}
\bigskip
\begin{abstract}
Investigators often use multi-source data (e.g., multi-center trials, meta-analyses of randomized trials, pooled analyses of observational cohorts) to learn about the effects of interventions in subgroups of some well-defined target population. Such a target population can correspond to one of the data sources of the multi-source data or an external population in which the treatment and outcome information may not be available. We develop and evaluate methods for using multi-source data to estimate subgroup potential outcome means and treatment effects in a target population. We consider identifiability conditions and propose doubly robust estimators that, under mild conditions, are non-parametrically efficient and allow for nuisance functions to be estimated using flexible data-adaptive methods (e.g., machine learning techniques). We also show how to construct confidence intervals and simultaneous confidence bands for the estimated subgroup treatment effects. We examine the properties of the proposed estimators in simulation studies and compare performance against alternative estimators. We also conclude that our methods work well when the sample size of the target population is much larger than the sample size of the multi-source data. We illustrate the proposed methods in a meta-analysis of randomized trials for schizophrenia.
\end{abstract}

\noindent%
{\it Keywords:} Causal inference; doubly robust; effect modifier; multi-source data; non-parametrically efficient; transportability; subgroup analyses.
\vfill

\newpage
\spacingset{1.45} 
\section{Introduction}
Many statistical inquiries, including quantifying the heterogeneity of treatment effects, can be better addressed by complex data from multiple sources with proper statistical methods. \citep{dahabreh2019generalizing, siddique2019causal, liu2022modeling, wang2020estimating, wang2023evaluating}
Multi-source data arise in multi-center trials, multi-cohort studies, and meta-analyses. Using multi-source data allows for more efficient estimation of heterogeneous treatment effects, a challenging task with single-source data due to limited sample size and restricted population reach \citep{kunzel2019metalearners}.

When using multi-source data, each data source often corresponds to a distinct, and sometimes ill-defined, population. However, decision-makers typically seek insights into the heterogeneity of a treatment effect in a well-defined target population so that the evidence from the analyses will be relevant to members of that well-defined target population. 
Such a target population may correspond to one data source of the multi-source data. For example, in a multi-cohort study, a target population can be represented by one source-specific of the cohorts; we call such a target population by a \textit{internal target population} \citep{dahabreh2019studydesigns}. In addition, it is also possible that decision-makers are interested in a population that is irrelevant to the population in the multi-cohort study, but a population where the treatment and outcome information may not be available (e.g., the treatment to be investigated is not accessible to the population). We call such a target population by an \textit{external target population} \citep{dahabreh2023efficient}.

Here, we focus on studying the \textit{subgroup treatment effects}, that is, conditional average treatment effects for subgroups of a target population. Such subgroups can be defined in terms of one or more covariates. Estimating subgroup treatment effects can support causal explanations, facilitate personalized treatment recommendations, and generate hypotheses for future studies \citep{oxman1992consumer, rothwell2005subgroup, wang2007statistics, sun2010subgroup, sun2014use, jaman2024penalized}. Recognizing these benefits, regulatory agencies have advocated reporting of estimated subgroup treatment effects in various disease domains \citep{ICH1998, FDA2018, FDA2019, FDA2023, amatya2021subgroup}. 

The rationale for utilizing multi-source data to estimate the subgroup treatment effects of an internal target population is that the large sample size of multi-source data can be helpful in improving the precision of the estimators. On the other hand, when estimating the subgroup treatment effects of an external target population, using multi-source data becomes essential for identifying the target parameter due to the lack of treatment and outcome information in the data collected for the external target population.

However, a challenge when estimating subgroup treatment effects using multi-source data is that one needs to adjust for covariate distribution imbalance across the various data sources \citep{dahabreh2023efficient}. Furthermore, when researchers aim to examine multiple subgroup effects concurrently, it becomes crucial to address the heightened risk of false positives arising from the multiple comparisons \citep{burke2015three}. Moreover, when the multi-source data are high-dimensional or complex, data-adaptive methods (e.g., machine learning techniques) may be most suitable to capture the complicated relationships among variables in estimating nuisance parameters. However, such flexible methods typically result in slower convergence rates unless estimators are carefully constructed to overcome this issue. 

In this paper, we develop doubly robust and non-parametric efficient estimators for subgroup treatment effects, where flexible data-adaptive methods can be employed to model the nuisance parameters. We also show how to construct simultaneous confidence bands for the estimators, which is suggested to be used by decision-makers when the interest lies in multiple subgroups simultaneously. Simulation studies and real data analyses validate and illustrate the advantages of our proposed estimators.

\section{Estimating subgroup treatment effects for an internal target population}\label{sec:internal}

In this section, we describe methods for estimating subgroup effects from multi-source data for an internal target population.

Suppose we have access to multiple datasets $\mathcal{S} = \{1, \ldots, m\}$ from different sources; each dataset, indexed by $s \in \mathcal{S}$, can be viewed as an internal target population. We further assume that these datasets comprise
simple random samples from (near-infinite) underlying superpopulations \citep{robins1988confidence}. For each observation, we have the information on the data source $S \in \mathcal{S}$, treatment $A\in\mathcal{A}$ that the patient received (with $\mathcal{A}$ assumed to be finite), baseline confounders $X\in\mathcal{X}$ (which may be high-dimensional), and outcome $Y$. Thus the observed data can be represented by $O_i=(Y_i, S_i, A_i, X_i)$, for $i=1,\ldots,n$, where $n$ is the sample size of the multi-source data. The sample size of each dataset is $n_s, \forall s=1,\dots,m$. The data structure is given in Figure \ref{fig:Internal}. 
\begin{figure}
\centering
\begin{subfigure}{.45\textwidth}
  \centering
  \includegraphics[scale=0.3]{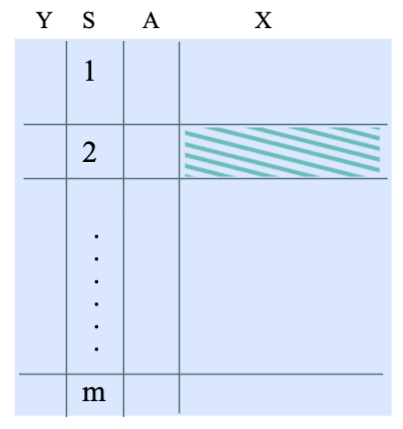}
  \caption{Data structure of the case described in Section \ref{sec:internal}, the shaded dataset represents an internal target population.}
  \label{fig:Internal}
\end{subfigure}%
\hfill
\begin{subfigure}{.45\textwidth}
  \centering
  \includegraphics[scale=0.26]{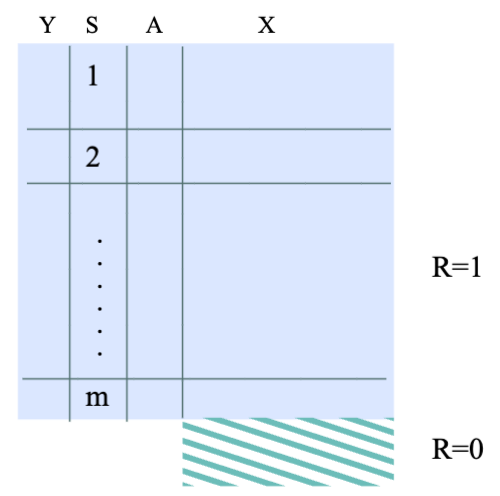}
  \caption{Data structure of the case described in Section \ref{sec:external}, the shaded dataset represents an external target population.}
  \label{fig:External}
\end{subfigure}
\caption{Data structures of two cases}
\end{figure}

Throughout this paper, we use a counterfactual framework \citep{robins2000d}. Define the counterfactual outcome $Y^a$ as the outcome that would be observed had the subject received treatment $a$. We are interested in investigating how conditional average treatment effects in an internal target population vary according to low-dimensional discrete effect modifiers $\widetilde{X}\subset X$. Therefore, in this section, the parameters of interest are subgroup potential outcome means specific to data source $s \in \mathcal{S}$, $\phi_{a, s}(\widetilde{x})=E(Y^{a}|\widetilde{X}=\widetilde{x}, S=s)$, and the corresponding subgroup treatment effect $\E(Y^{a}-Y^{a'}|\widetilde{X}=\widetilde{x}, S=s)= \phi_{a, s}(\widetilde{x})-\phi_{a', s}(\widetilde{x})$.
Note that these parameters are based on counterfactual outcomes, which are not fully observed. We thus first provide sufficient identifying conditions to allow estimation of these parameters.
\subsection{Identification}
We consider the following assumptions that are sufficient to identify $\phi_{a, s}(\widetilde{x})$:

\noindent
\textit{A1. Consistency of potential outcomes:} if $A_i = a,$ then $Y^a_i = Y_i$, for every individual $i$ and every treatment $a \in \mathcal A$. \\
This assumption encodes two requirements: (1) no interference, i.e., outcomes are not affected by treatment of other subjects \citep{vanderWeele2009}, (2) a single version of the treatment of interest \citep{rubin2010reflections}.
\\
\noindent
\textit{A2. Exchangeability over treatment in an internal target population:} for each $a \in \mathcal A$ and an internal target population $s$, $Y^a \indep A | (X, S=s)$. \\
Equivalently, this condition assumes that $X$ can sufficiently adjust for confounding in each internal target population.\\
\noindent
\textit{A3. Positivity of the probability of treatment in each of internal target population:} for each treatment $a \in \mathcal A$, if $f(x, S=s) \neq 0$, then $\Pr(A = a | X = x, S = s) > 0$.\\
\noindent
\textit{A4. Exchangeability over source index:} for each $a \in \mathcal A$, $Y^a \indep S | X$. \\
Like Assumption \textit{A2}, the plausibility of Assumption \textit{A4} would need to be assessed based on substantive knowledge. In practice, the impact of violations of these conditions may be evaluated by sensitivity analyses \citep{robins2000sensitivity}. It is also worth noting that the combination of Assumption \textit{A2} and \textit{A4} implies that for each $a \in \mathcal A$, $Y^a \indep (S, A) | X$ and thus $Y \indep S | (X, A=a)$.\\
\noindent
\textit{A5. Positivity of the probability of participation in the internal target population:} $\forall x\in\mathcal{X}, \forall s \in\mathcal{S}, \Pr(S = s | X = x) >0$ with probability 1.\\
Because $\widetilde{X}\subset X$, the above assumption implies that $\Pr(\widetilde{X}=\widetilde{x}, S = s)=\Pr(S = s|\widetilde{X}=\widetilde{x})\Pr(\widetilde{X}=\widetilde{x})>0$. 
That is, the probability of observing the subgroup in the internal target population data is positive, i.e., the effect modifier patterns of interest are observable in the internal target population data. 

When Assumptions \textit{A1} through \textit{A5} hold, our parameter of interest $\phi_{a, s}(\widetilde{x})$ can be identified using the observed data $O$.

\begin{theorem}[Identification of subgroup potential outcome means in an internal target population]
\label{thm:identification2}
Under conditions \textit{A1} through \textit{A5}, the subgroup potential outcome means in an internal target population under treatment $a \in \mathcal A$, $\phi_{a, s}(\widetilde{x})=\E(Y^a | \widetilde{X}=\widetilde{x}, S=s)$, is identifiable by the observed data functional
\begin{equation*} 
  \phi_{a, s}(\widetilde{x}) \equiv 
  \E\big\{\E(Y |A=a, X) | \widetilde{X}=\widetilde{x}, S=s\big\},
\end{equation*}
which can be equivalently expressed as 
\begin{equation*} 
  \phi_{a, s}(\widetilde{x})=\dfrac{1}{\Pr(\widetilde{X}=\widetilde{x}, S=s)}\E\Big\{\dfrac{I(A=a, \widetilde{X}=\widetilde{x})Y \Pr(S=s|X)}{\Pr(A=a|X)}\Big\}.
\end{equation*}
\end{theorem}
The proof is given in \ref{appendix:identification2}.
It immediately follows that under the conditions of Theorem~\ref{thm:identification2}, subgroup treatment effects are identifiable: $\E(Y^{a}-Y^{a'}|\widetilde{X}=\widetilde{x}, S=s)=\phi_{a, s}(\widetilde{x})-\phi_{a', s}(\widetilde{x})$. In addition, this contrast remains identifiable when Assumption \textit{A4} is replaced by a slightly weaker condition: $(Y^{a}-Y^{a'})\indep S|X$.
\subsection{Estimation}\label{sec:estimation2}
One can use either of the two identifying formulae from Theorem~\ref{thm:identification2} to estimate $\phi_{a, s}(\widetilde{x})$ by plugging in estimators of the unknown nuisance functions and replacing expectations with sample averages. However, consistency of these approaches would require consistent models for the component nuisance functions. If parametric models are chosen, it is quite likely in practice that these may be misspecified, especially if the data are complex or high-dimensional. 
Therefore, non-parametric models for the nuisance functions may be preferable, though these would result in slow convergence rates for the two plug-in estimators described.

To achieve fast convergence and efficient estimation of the subgroup effect parameter $\phi_{a, s}(\widetilde{x})$, while allowing flexible nuisance models, we will base estimation on the influence function of $\phi_{a, s}(\widetilde{x})$. In \ref{appendix:influence_function2}, we show that under the nonparametric model for the observable data, the first-order influence function \citep{bickel1993efficient} of $\phi_{a, s}(\widetilde{x})$ is
\begin{align*}
\mathit\Phi_ {p_{0}}(a, s, \widetilde{x})=&\dfrac{1}{\Pr_ {p_{0}}(\widetilde{X}=\widetilde{x}, S=s)}
    \Bigg[I(\widetilde{X}=\widetilde{x}, S=s)\big\{\E_{p_{0}}(Y | A=a, X)-\phi_{a, s}(\widetilde{x})\big\}\\
    +&I(A=a, \widetilde{X}=\widetilde{x})\dfrac{\Pr_{p_{0}}(S=s|X)}{\Pr_{p_{0}}(A=a|X)}\Big\{Y-\E_{p_{0}}(Y | A=a, X)\Big\}\Bigg],
\end{align*}
where the subscript $p_{0}$ denotes that all quantities are evaluated at the ``true'' data law. We also show that it is the efficient influence function under certain semiparametric models. The proof of the following result is given in \ref{appendix:corollary}.

\begin{restatable}[Efficient influence function under semiparametric model]{corollary}{cor:identification2}
\label{cor:identification2}
The efficient influence function under the semi-parametric model which incorporates the constraint $Y \indep S | (X, A=a)$, or/and knowing the propensity score $\Pr(A=a|X, S=s)$ (e.g., the multi-source data are a collection of randomized clinical trials) is $\mathit\Phi_ {p_{0}}(a, s, \widetilde{x})$.
\end{restatable}
\noindent
According to Corollary \ref{cor:identification2}, efficient influence function-based estimation and inference will be the same if we incorporate either of the stated restrictions into the model.
The influence function above suggests the estimator
\begin{equation*}
  \widehat \phi_{a, s}(\widetilde{x}) =
  \dfrac{\widehat \kappa(\widetilde{x})}{n}\sum\limits_{i=1}^{n} \Bigg[ I(\widetilde{X}_{i}=\widetilde{x}_{i}, S_i = s) \widehat \mu_a(X_i) 
  +I(A_i = a, \widetilde{X}_{i}=\widetilde{x}_{i}) \dfrac{\widehat q_{s}(X_i)}{\widehat \eta_a(X_i)}  \Big\{ Y_i - \widehat \mu_a(X_i) \Big\} \Bigg],
\end{equation*}
where $\widehat \kappa(\widetilde{x}) = \{n^{-1} \sum_{i=1}^n I(\widetilde{X}=\widetilde{x}, S_i=s)\}^{-1}$ is an estimator for $\kappa=\Pr(\widetilde{X}=\widetilde{x}, S=s)^{-1}$, $\widehat \mu_a(X)$ is an estimator for $\mu_a(X)=\E(Y | A=a, X)$, $\widehat \eta_a(X)$ is an estimator for $\eta_a(X)=\Pr(A = a| X)$, and $\widehat q_{s}(X)$ is an estimator for $q_{s}(X)=\Pr(S = s | X)$. Note that one can use the decomposition $\sum_{s=1}^{S}\Pr(A=a|X, S=s)\Pr(S=s|X)$ to estimate $\eta_a(X)$.
 
For simplicity, we assume throughout that the nuisance estimates are trained in a separate independent sample, so as not to rely on empirical process conditions (e.g., Donsker-type conditions). In practice, with only one independent identical distributed sample, the same goal can be achieved by the use of sample splitting and cross-fitting, i.e., splitting the data into two halves (each half consists of stratified samples according to data source), using one for fitting nuisance estimates and the other for computing $\widehat \phi_{a, s}(\widetilde{x})$ as above \citep{schick1986, pfanzagl2012}. One can then swap samples, do the same procedure, and then average to obtain the final estimate \citep{chernozhukov2018double}. Note that this cross-fitting procedure can easily be extended to $> 2$ folds (e.g., see \citet{chernozhukov2018double, kennedy2020sharp}), but we will focus on the case of one split in this paper. When stratified sample splitting is not practical (e.g., small sample-sized data source), to achieve the asymptotic properties demonstrated below, one can use ensemble bagged estimators \citep{chen2022debiased}. Finally, in practice, we suggest users split the data stratified by the data source to make sure each split sample contains all sources of data.
\subsection{Asymptotic theory}
In this section, we give the properties of our proposed estimator. In particular, we show that our estimator is doubly robust, and under certain conditions, the estimator is asymptotically normal and non-parametrically efficient. 
Notably, the properties do not rely on the assumption that the nuisance models fall into Donsker classes \citep{vandervaart2000asymptotic}, which enables users to utilize flexible machine-learning methods to adjust for high-dimensional covariates.
\subsubsection{Asymptotic properties of the pointwise estimator}
We first notice that  $\widehat \kappa$ satisfies i) $I(\widehat \kappa^{-1}=0)=o_p(n^{-1/2})$, and ii) $\widehat \kappa\xrightarrow{P} \kappa$, see proof in \ref{app:lemma1}.
\noindent
To establish asymptotic properties of $\widehat \phi_{a, s}(\widetilde{x})$, we make the following assumptions:
\begin{enumerate}
\item[(a1)] $\exists \varepsilon>0, \quad s.t. \quad 
\Pr\{\varepsilon\leqslant \widehat \eta_a(X)\leqslant1-\varepsilon\}=
\Pr\{\varepsilon\leqslant \widehat q_s(X)\leqslant1-\varepsilon\}=1$
\item[(a2)] $ \E(Y^{2})<\infty$,
\item[(a3)] $||\widehat \mu_a(X) -\mu_a(X)||+||\widehat \eta_a(X) -\eta_a(X)||+||\widehat q_s(X) -q_s(X)||=o_p(1)$,
\end{enumerate}
where $||\cdot||$ denotes the $L_2$ norm, i.e., $\lVert f \rVert^2 = \E(f^2)$, for any $f \in L_2(P)$. 
The following Theorem gives the asymptotic properties of $\widehat \phi_{a, s}(\widetilde{x})$; a detailed proof is given in \ref{app:Proof-Estimation2}. 
 
\begin{restatable}{theorem}{thmestimationSA}
\label{thm:estimation2}
If assumptions \textit{A1} through \textit{A5}, and \textit{(a1)} through \textit{(a3)} hold, then 
\begin{align*}
     \widehat \phi_{a, s}(\widetilde{x})-\phi_{a, s}(\widetilde{x})
    = \mathbb{P}_n\{\mathit\Phi_ {p_{0}}(a, s, \widetilde{x})\}+
    R_{n}+
    o_p(n^{-1/2}),
\end{align*}
where $R_{n}\lesssim O_p\big\{\widehat \mu_a(X) -\mu_a(X)||(||\widehat \eta_a(X) -\eta_a(X)||+||\widehat q_s(X) -q_s(X)||)\big\}$, and nuisance parameters are estimated on a separate independent sample.

In particular, if $R_{n}= o_p(n^{-1/2})$, then 
$
    \sqrt{n}\big\{\widehat \phi_{a, s}(\widetilde{x})-\phi_{a, s}(\widetilde{x})\big\}\rightsquigarrow\mathcal{N}\Big[0, \E_{p_{0}}\big\{\mathit\Phi_ {p_{0}}(a, s, \widetilde{x})^{2}\big\}\Big].
$
That is, $\widehat \phi_{a, s}(\widetilde{x})$ is non-parametric efficient.
\end{restatable}
\noindent
Theorem \ref{thm:estimation2} characterizes the large-sample behavior of the proposed estimator. With sample splitting, the usual Donsker class assumptions for estimating the nuisance parameters are avoided. Instead, assumption (a3) only requires that the nuisance functions $\mu_a(X)$, $\eta_a(X)$, or $q_s(X)$ are consistently estimated---though potentially quite slowly---in $L_2$ norm. Many flexible and data-adaptive machine learning methods satisfy the requirement, which creates opportunities to adjust high-dimensional covariates flexibly.

Under assumptions (a1) through (a3), the proposed estimator is consistent with convergence rate $O_p(n^{-1/2}+R_{n})$. By the product structure of $R_n$, the estimator is doubly robust: if either $\mu_a(X)$ or $\{\eta_a(X), q_s(X)\}$ is estimated consistently (converge at rate $o_{p}(1)$), then the estimator is consistent, i.e., $||\widehat \phi_{a, s}(\widetilde{x})-\phi_{a, s}(\widetilde{x})||=o_{p}(1)$ \citep{kennedy2019robust}. 

The theorem also lays out the conditions for the estimator to be $\sqrt{n}$-consistent, asymptotically normal and semiparametric efficient; it is sufficient that $R_{n} = o_p(n^{-1/2})$. Note that this condition can hold even if all nuisance functions are estimated at slower than parametric $\sqrt{n}$-rates, so that efficiency of  estimation and validity of inference may be justified even when machine learning methods are used to model the nuisance functions. For instance, if all the nuisance functions are estimated at faster than $n^{1/4}$ rates, i.e., $\left\lVert \widehat\mu_a(X) - \mu_a(X)\right\rVert =\left\lVert \widehat\eta_a(X) - \eta_a(X)\right\rVert = \left\lVert \widehat q_s(X) - q_s(X)\right\rVert = o_p(n^{-1/4})$, then the condition holds. This is possible under sparsity, smoothness, or other structural assumptions. For example, under some conditions \citep{horowitz2009semiparametric}, generalized additive model estimators can attain rates of the form $O_p(n^{-2/5})$, which is $o_p(n^{-1/4})$. In general, the convergence rate for $\widehat \phi_{a, s}(\widetilde{x})$ can be much faster than that of any component nuisance function. \citep{kennedy2019robust} If the condition $R_n = o_p(n^{-1/2})$ is satisfied, then one can construct pointwise confidence intervals by either the bootstrap or a direct estimate of the asymptotic variance. 
\subsubsection{Simultaneous confidence bands}\label{sec:Simultaneous}
In the case that there are many subgroups of interest, i.e., $\widetilde{X}$ has many categories, the 95\% pointwise confidence intervals of each $\widehat \phi_{a, s}(\widetilde{x})$ do not reveal a valid 95\% confidence region for all the estimates simultaneously.
Usually, a valid simultaneous confidence band is wider than pointwise confidence intervals. Pursuing such uniform bands is of direct interest when we are simultaneously interested in all subgroup effects.

We show in Corollary \ref{lemma:simultaneous confidence bands_bootstrap} that (asymptotic) simultaneous confidence bands of $\widehat \phi_{a, s}(\widetilde{x})$ can be constructed via the bootstrap. Proofs are given in \ref{appendix:simultaneousconfidencebands_bootstrap}.
\begin{restatable}{corollary}{}
\label{lemma:simultaneous confidence bands_bootstrap}
Let $\widehat \phi_{a, s}(\widetilde{x})$ and $\widehat\sigma_{a,s}(\widetilde{x})$ be the estimated $\phi_{a, s}(\widetilde{x})$ and its standard deviation respectively.
Suppose $\widetilde{x}$ is a vector with support $\widetilde{X}$, and let $$
t^b_{max}=\text{sup}_{\widetilde{x}\in\widetilde{X}}\Big|\dfrac{\widehat \phi_{a, s}^b(\widetilde{x})-\widehat \phi_{a, s}(\widetilde{x})}{\widehat\sigma_{a,s}(\widetilde{x})}\Big|,
$$
where $\widehat \phi_{a, s}^b(\widetilde{x})$ is the point estimate of $\widehat \phi_{a, s}(\widetilde{x}), b=1,\dots,B$ from the $b$-th bootstrapped data. If assumptions (a1) through (a3) hold, then the $(1-\alpha)$ simultaneous confidence bands of $\widehat \phi_{a, s}(\widetilde{x})$ is $\{\widehat \phi_{a, s}(\widetilde{x})-c(1-\alpha)\widehat \sigma_{a,s}(\widetilde{x}),\text{ } \widehat \phi_{a, s}(\widetilde{x})+c(1-\alpha)\widehat \sigma_{a,s}(\widetilde{x})\}$, where the critical value $c(1-\alpha)$ is the $(1-\alpha)$-quantile empirical distribution of $t^b_{max}$ over $B$ replicates.
\end{restatable}
\noindent
We also show in \ref{app:SCB_another} that when certain conditions hold, the confidence bands can be approximated by another expression which avoids computing $t^b_{max}$. We verified empirically that when the sample size $n$ and number of bootstrap replications $B$ are sufficiently large, these two approaches result in numerically similar confidence bands.
\section{Estimating subgroup treatment effects for an external target population}\label{sec:external}


Next, we present the approach to estimate the subgroup treatment effects for an external target population. Similar to the data structure introduced in Section \ref{sec:internal}, we assume multi-source data $\mathcal{S}$ are collected with information of outcome $Y$, treatment assignment $A\in\mathcal{A}$ and (high-dimensional) covariates $X$. The observed data is given by  by $O_i=(Y_i, X_i, A_i, S_i)$, for  $i=1,\ldots,n_1+\ldots+n_m$. Different from the last section, we have access to the same (high-dimensional) covariates information from the external population of interest. We assume the external data collected are $n_0$ simple random samples from the population of interest. The observed data from the external target population can be presented as $O_j=(X_j)$, for $j=1,\ldots,n_0$. We use the random variable $R$ to indicate belonging in the multi-soure data. Therefore, combining the two datasets gives us the following structure $O_i=\{Y_{i}, A_{i}, X_{i}, S_{i}, R_{i}\}$, for $i=1,\ldots,n$, where $n=n_{0}+n_{1}+\ldots+n_{m}$ with  $Y_{i}=A_{i}=S_{i}=NA, \forall i\in\{i: R_i=0\}$. We abuse the notation a bit and denote $n_{0}+n_{1}+\ldots+n_{m}$ by $n$ for the rest of this section (whereas in the last section, $n$ is defined as $n_{1}+\ldots+n_{m}$). We are interested in estimating $\psi_a(\widetilde{x})=E(Y^{a}|\widetilde{X}=\widetilde{x}, R=0)$. The data structure is given in Figure \ref{fig:External}.

The assumptions for identifying the subgroup potential outcome mean in the external data are the same as assumptions \textit{A1-A5}, with an additional assumption\\
\textit{A6. Exchangeability over external data indicator:} for each $a\in\mathcal{A}, Y^a\indep R|X.$\\  
Next, we give the identification of the parameter of interest; the proof is given in \ref{appendix:identification1}. 

\begin{restatable}[Identification of subgroup potential outcome means in an external target population]{theorem}{thmidentificationcollection}
\label{thm:identification1}
Under conditions A1 through A6, the subgroup potential outcome means in the external population under treatment $a \in \mathcal A$, $\psi_a(\widetilde{x})=\E(Y^a | \widetilde{X}=\widetilde{x}, R = 0 )$, is identifiable by the observed data functional
$
  \psi_a(\widetilde{x}) \equiv 
  \E\big\{\E(Y |A=a, X, R=1) | \widetilde{X}=\widetilde{x}, R=0\big\},
$
which can be equivalently expressed as 
$
  \psi_a(\widetilde{x}) = \dfrac{1}{\Pr(\widetilde{X}=\widetilde{x}, R=0)} \E\Big\{\frac{I(A=a, \widetilde{X}=\widetilde{x}, R=1) Y \Pr(R=0|X)}{\Pr(R=1|X) \Pr(A=a|X, R=1)} \Big\}.
$
\end{restatable}
\noindent

Similar to Section \ref{sec:estimation2}, in \ref{appendix:influence_function1} we show the first-order influence function of $\psi_a(\widetilde{x})$, for $a \in \mathcal A$. We also show in \ref{app:cor:sameIF_phi} that the above influence function is also the efficient influence function under the semi-parametric model, which incorporates the constraint $Y \indep S | (X, A=a)$, or/and knowing the propensity score $\Pr(A=a|X, S=s)$. 
The influence function above suggests the estimator
{\small
\begin{equation*}
  \widehat \psi_a(\widetilde{x}) =
  \dfrac{\widehat \gamma(\widetilde{x})}{n}\sum\limits_{i=1}^{n} \Bigg[ I(\widetilde{X}_{i}=\widetilde{x}_{i}, R_i = 0) \widehat g_a(X_i) 
  +
  I(A_i = a, \widetilde{X}_{i}=\widetilde{x}_{i}, R_i=1) \dfrac{1 - \widehat p(X_i)}{\widehat p(X_i) \widehat e_a(X_i)}  \Big\{ Y_i - \widehat g_a(X_i) \Big\} \Bigg],
\end{equation*}}
where $\widehat \gamma(\widetilde{x}) = \{n^{-1} \sum_{i=1}^n I(\widetilde{X}=\widetilde{x}, R_i = 0)\}^{-1}$ is an estimator for $\gamma=\Pr
 (\widetilde{X}=\widetilde{x}, R=0)^{-1}$, $\widehat g_a(X)$ is an estimator for $g_a(X)=\E(Y | A=a, X, R = 1)$, $\widehat e_a(X)$ is an estimator for $e_a(X)=\Pr(A = a| X, R = 1)$, and $\widehat p(X)$ is an estimator for $p(X)=\Pr(R = 1 | X)$.  
We propose to use the same strategy, sample splitting that is described in Section \ref{sec:estimation2} to estimate the target parameter $\psi_a(\widetilde{x})$.

Similar to Lemma~\ref{lemma:gamma}, it can be proved that $I(\widehat \gamma(\widetilde{x})^{-1}=0)=o_p(n^{-1/2})$ and $\widehat \gamma(\widetilde{x})\xrightarrow{P} \gamma(\widetilde{x})$.
To establish asymptotic properties of $\widehat \psi_a(\widetilde{x})$, we make the following assumptions:
\begin{enumerate}
\item[(b1)] $\exists \varepsilon>0, \quad s.t. \quad 
\Pr\{\varepsilon\leqslant \widehat e_a(X)\leqslant1-\varepsilon\}=
\Pr\{\varepsilon\leqslant \widehat p_a(X)\leqslant1-\varepsilon\}=1$
\item[(b2)] $ \E(Y^{2})<\infty$,
\item[(b3)] $||\widehat g_a(X) -g_a(X)||+||\widehat e_a(X) -e_a(X)||+||\widehat p_a(X) -p_a(X)||=o_p(1)$. 
\end{enumerate}
The following Theorem gives the asymptotic properties of $\widehat \psi_a(\widetilde{x})$; a detailed proof is given in \ref{app:Proof-Estimation1}. 
 
\begin{restatable}{theorem}{thmestimationSA}
\label{thm:estimation1}
If assumptions \textit{A1} through \textit{A6}, and \textit{(b1)} through \textit{(b3)} hold, then 
\begin{align*}
     \widehat \psi_a(\widetilde{x})-\psi_a(\widetilde{x})
    =\mathbb{P}_n\{\mathit\Psi_ {p_{0}}(a, \widetilde{x})\}+
    Q_{n}+o_p(n^{-1/2}),
\end{align*}
where $Q_{n}\lesssim O_p\big\{||\widehat g_a(X) -g_a(X)||(||\widehat e_a(X) -e_a(X)||+||\widehat p_a(X) -p_a(X)||)\big\}$, and nuisance parameters are estimated on a separate independent sample.\\
In particular, if $Q_{n} = o_p(n^{-1/2})$, then 
$
    \sqrt{n}\big\{\widehat \psi_a(\widetilde{x})-\psi_a(\widetilde{x})\big\}\rightarrow\mathcal{N}\Big[0, \E_{p_{0}}\big\{\mathit\Psi_ {p_{0}}(a, \widetilde{x})^{2}\big\}\Big].
$\\
That is, $\widehat \psi_a(\widetilde{x})$ is non-parametric efficient.
\end{restatable}
\noindent
We can also construct the simultaneous confidence bands of the estimator with the same approach described in Section \ref{sec:Simultaneous}.

\section{Simulation Studies}

In this section, we summarize the results of two simulation studies we conducted to evaluate the proposed estimators. In the first study, we assessed the finite sample bias, standard error, and confidence interval coverage of the proposed estimators illustrated in the previous sections, comparing them with alternative estimators. In the second study, we verify the rate robustness properties suggested by Theorem~\ref{thm:estimation2} (and analogously Theorem~\ref{thm:estimation1}).
\subsection{Assessment of consistency and efficiency}
We generate data that contain both multi-source data and external data. The total sample size is either $n=\sum_{s=0}^{m}n_s=10^4 \text{ or }10^5$. 

We first generate a five-level categorical variable $X_1=\widetilde{X}$ by categorizing a standard normal distributed variable into five levels (with values 1 to 5 respectively) with the sample size ratio 1:2:3:2:1. Then generate nine variables $X_2,..., X_{10}$ from a mean 0.1 multivariate normal distribution with all marginal variances equal to 0.25 and all pairwise correlations equal to 0.5. Therefore, the covariates matrix $X^{n\times p}$ contains $p=10$ variables.

We consider the multi-source data come from three sources, $\mathcal{S}=\{1, 2, 3\}$, where the sample size of the multi-source data $\sum_{s=1}^{3}n_s$ can be 1000, 2000 or 5000. We assume the multi-source data indicator follows a Bernoulli distribution, $R\sim \text{Bernoulli}\{\Pr(R=1|X)\}$ with $\Pr(R=1|X)=\dfrac{\text{exp}(\beta X^T) }{1+ \text{exp}(\beta X^T)}$, $X = (1, X_{1}, \ldots, X_{10})$ and $\beta$ is a vector size of 11 containing $\beta_0$ and all others $\ln(1.05)$, where we solved for $\beta_0$ to result (on average) in the desired total number of multi-source data sample size \citep{robertson2021intercept}. In the multi-source data, we allocate the subjects to three internal source-specific data using a multinomial logistic model. That is, we assume $S|(X, R=1)\sim \text{Multinomial}(p_1, p_2, p_3; \sum_{s=1}^{3}n_s)$ with $p_1 = \Pr(S= 1 | X, R = 1 ) = \dfrac{e^{\xi X^T} }{1+ e^{\xi X^T} + e^{\zeta X^T} }$, $p_2 =\Pr(S= 2| X, R = 1 ) = \dfrac{e^{\zeta X^T} }{1+ e^{\xi X^T} + e^{\zeta X^T} }$, and $p_3 = \Pr(S= 3 | X, R = 1 ) = 1 - p_1 - p_2$ where $\xi$ and $\zeta$ are vectors size of 11 containing $\xi_0$ and $\zeta_0$ and all others range from $\ln(1.1)$ to $\ln(1.5)$. We used Monte Carlo methods to obtain intercepts $\xi_0$ and $\zeta_0$ that resulted in the sample sizes $n_1:n_2:n_3= 4:2:1$. The sample size of the external data $\{R=0\}, n_0=n-\sum_{s=1}^{3}n_s$.

We assume the treatment assignment follows a Bernoulli distribution, $A\sim \text{Bernoulli}\{\Pr(A=1|X, S)\}$ with $\Pr(A=1|X, S)=\dfrac{e^{\alpha_s X^T} }{1+ e^{\alpha_s X^T} }$, where $\alpha_s$ ranges from -2 to 1.5. We do not need to generate $A$ when $R=0$.
We generated potential outcomes as $Y^a =1+5a+\sum_{j=2}^{4}f_1(\theta X_j)+\sum_{j=5}^{7}f_2(\theta X_j)+\sum_{j=8}^{10}f_3(\theta X_j)+\sum_{k=1}^{5}\vartheta_k X_{1k} a+\sum_{j=2}^{5}\iota_j X_j a+ \epsilon^a$, where $f_1(x)=\text{sin}(x)/5, f_2(x)=\text{exp}(-0.25x), f_3(x)=0.02x^2+(2+0.2x)^2+2(0.015x)^3, \theta=0.75, \vartheta_k^1=(0.2, 0.4, -0.5, 0.1, -0.01), k=1,...,5, \iota_j^1=(0.2, 0.2, -0.2, -0.2), j=2,...,5$. In all simulations, $\epsilon^a$ had an independent normal distribution with mean 0 variance 10 for $a=0,1$. We generated observed outcomes under consistency, such that $Y = AY^1 + (1 - A) Y^0$.
   
We use logistic regression to estimate $\widehat p(X)$,  and estimate $\widehat \eta_{a}(X)$ and $\widehat e_{a}(X)$ by multiplying two probabilities obtained from logistic and multinomial regressions, multinomial model to estimate $\widehat q_{s}(X)$, and generalized additive model to estimate $\widehat \mu_{a}(X) \text{ and } \widehat g_{a}(X)$. To misspecify the model, we omitted $X_j, j=5,\dots,10$ \citep{kang2007demystifying}. In addition, we use linear regression to misspecify $\widehat \mu_{a}(X) \text{ and } \widehat g_a(X)$ excluding the interaction terms of $X$ and $A$.

We evaluate the performance of the proposed estimators for all the five subgroups. For comparison, we compare the proposed doubly robust (DR) estimator with the G computation (Regression) 
$
\widehat\phi_{a,s}^g(\widetilde{x})=\widehat \kappa(\widetilde{x})/n \cdot \sum_{i=1}^{n}I(\widetilde{X}_{i}=\widetilde{x}_{i}, S_i = s) \widehat \mu_a(X_i)$,
 and inverse probability of treatment weighting (IPTW)
$\widehat\phi_{a,s}^w(\widetilde{x})=\widehat \kappa(\widetilde{x})/n \cdot\sum\limits_{i=1}^{n} \big[ I(A_i = a, \widetilde{X}_{i}=\widetilde{x}_{i})\widehat q_{s}(X_i)/\widehat \eta_a(X_i)  \cdot Y_i\big].
$ Similar for $\widehat\psi_{a}(\widetilde{x})$.
All simulation results are based on 5000 runs.

Figure \ref{fig:Bias_SD} shows the bias, coverage of simultaneous confidence interval and pointwise confidence interval, theoretical standard deviation (estimated by influence functions) and Monte-Carlo standard deviation of $\widehat \phi_{1,1}(\widetilde{x}), \widetilde{x}=1,...,5$ when the sample size of the multi-source data is 1000 ($\sum_{s=1}^{3}n_s=1000$). The 95\% simultaneous confidence interval in Figure \ref{fig:Bias_SD} is obtained by Corollary \ref{lemma:simultaneous confidence bands}, but it is checked that numerically, it is almost the same as the one obtained from Corollary \ref{lemma:simultaneous confidence bands_bootstrap}. For each subgroup, the bias is near zero, the theoretical standard deviation and Monte-Carlo standard deviation are almost the same (with the latter one numerically larger than the former one but negligible in their scales), and the coverage of simultaneous confidence interval is uniformly higher than the coverage of pointwise confidence interval (at around 95\%). 
The corresponding simulation results for $\widehat \psi_{1}(\widetilde{x}), $ when the sample size $n=10^4$ and $\sum_{s=1}^{3}n_s=1000$ is given in \ref{app:add_sim} Figure \ref{app:fig:r1}, which shows similar results. When the sample size $n=10^5$ and/or $\sum_{s=1}^{3}n_s=2000$ or $5000$, simulations give similar results, which are given in \ref{app:add_sim} Table \ref{app:table:add_s}. 

\begin{figure}
    \centering
    \includegraphics[width=0.5\textwidth]{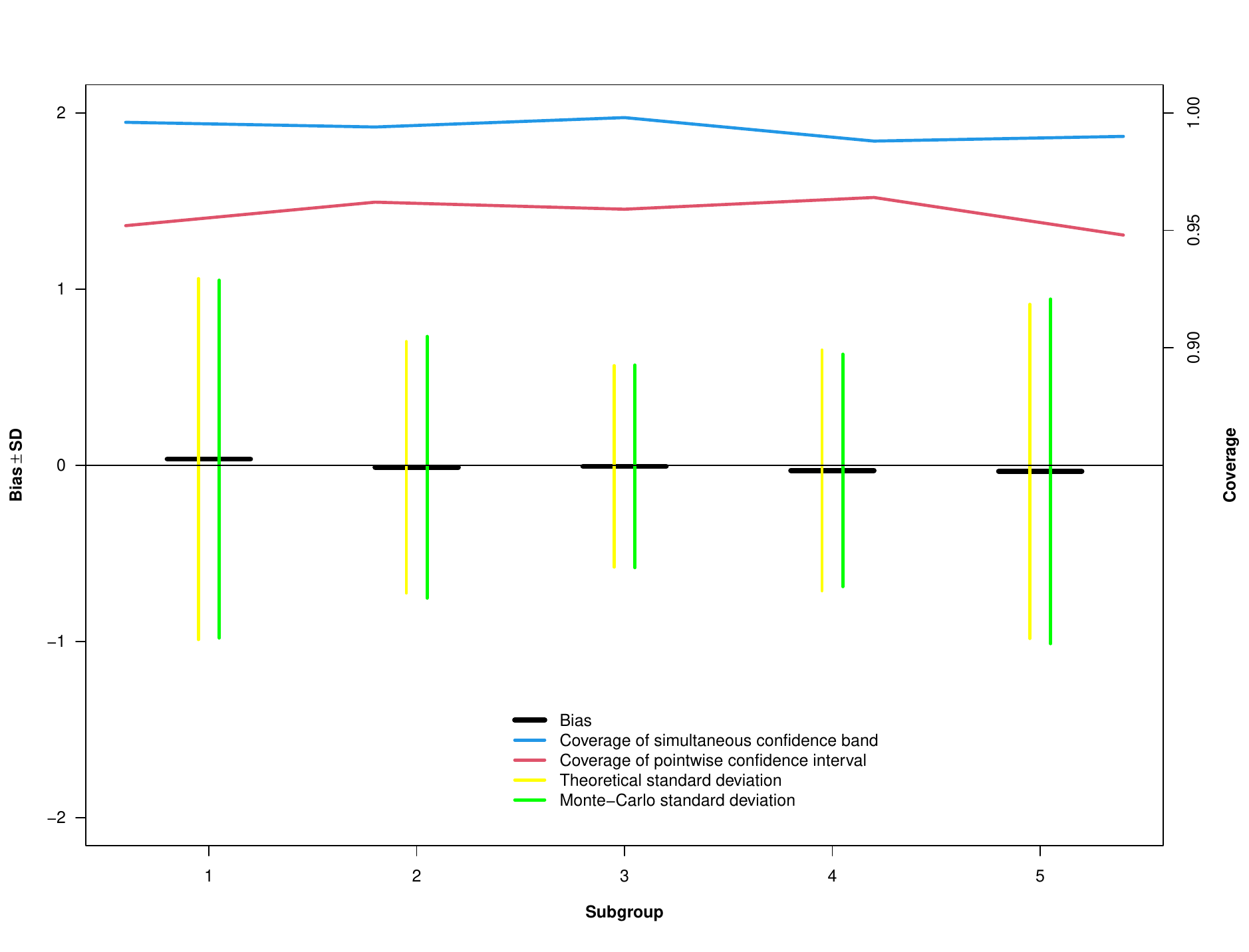}
    \caption{Biases, coverage of simultaneous confidence bands, coverage of pointwise confidence intervals, theoretical standard deviation and Monte-Carlo standard deviation of $\widehat \phi_{1}(\widetilde{x}), \widetilde{x}=1, 2, 3, 4, 5$. Sample size is 1000 with 5000 iterations.}
    \label{fig:Bias_SD}
\end{figure}
Figure \ref{fig:model_comparison} shows the averaged biases and Monte-Carlo standard errors of various estimators with different combination(s) of correctly specified models for estimating $\widehat \phi_{1,1}(3)$ when the sample size $\sum_{s=1}^{3}n_s$ is 1000. It is verified that our proposed estimator is doubly robust.
The corresponding simulation results for $\widehat \psi_{1}(3)$ when the sample size $n=10^4$ and $\sum_{s=1}^{3}n_s=1000$ is given in \ref{app:add_sim} Figure \ref{app:fig:r2}, which shows similar results.
\begin{figure}
    \centering
    \includegraphics[width=0.5\textwidth,trim={0 0 0 2cm},clip]{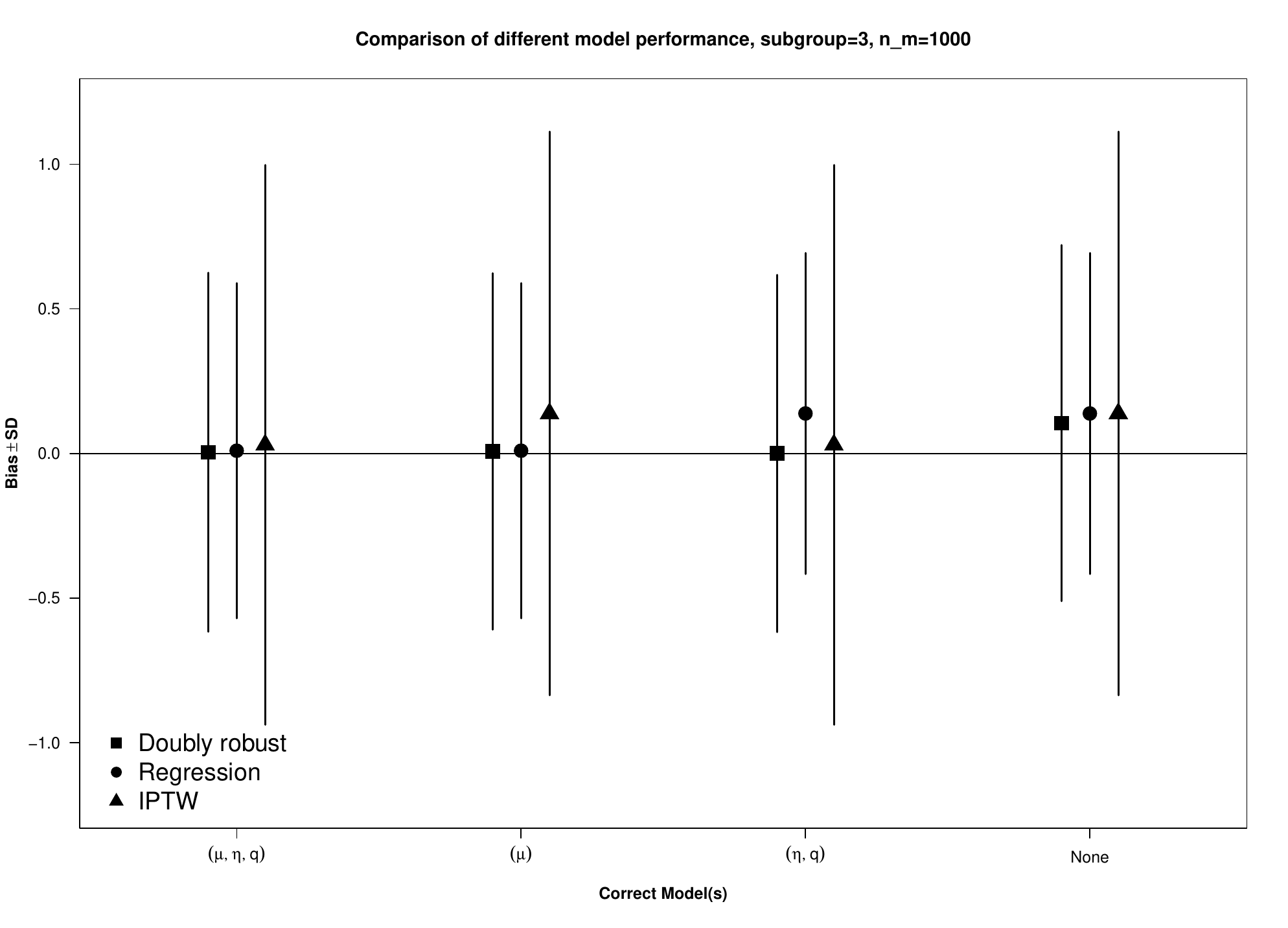}
    \caption{Simulation results for $\widehat \phi_{1,1}(3)$ when the sample size $\sum_{s=1}^{3}n_s$ is 1000 with 5000 iterations. The results include the averaged biases and Monte-Carlo standard errors of various estimators with different combination(s) of correctly specified models. The methods being compared are the proposed doubly robust (DR) estimator $\widehat\phi_{1,1}(3)$, the estimator that only uses outcome regression (Regression) $\widehat\phi^g_{1,1}(3)$, and the inverse probability of treatment weighting (IPTW) $\widehat\phi^w_{1,1}(3)$.}
    \label{fig:model_comparison}
\end{figure}
We additionally conduct simulations for different sample size settings, the results are given in \ref{app:add_sim} Tables \ref{app:table:add_r1} and \ref{app:table:add_r2}. It can be shown that the biases and standard deviations decrease with the increase in sample size. But when the model(s) is/are incorrectly specified, the biases of the corresponding estimator do not converge to zero.
Notably, \citet{schmid2020comparing} pointed out that some statistical methods performed poorly when generalizing effect estimates from randomized controlled trials to much larger target populations (which is a special case of transporting the effects to an external data source). From the simulation results, it can be verified that our methods work well when the sample size of the multi-source data is only 1\% of the sample size of the target population. 

\subsection{Assessment of rate robustness}

Theorem~\ref{thm:estimation2} shows that if $\lVert \widehat{\mu}_a - \mu_a\rVert \left\{\lVert \widehat{\eta}_a - \eta_a\rVert + \lVert \widehat{q}_s - q_s\rVert\right\} = o_p(n^{-1/2})$, then $\widehat{\phi}_{a,s}(\widetilde{x})$ is $\sqrt{n}$-consistent and asymptotically efficient. More generally, the asymptotic product-bias described in this result yields a distinct advantage over simpler plug-in or inverse-probability weighted estimators whose asymptotic bias is typically first order.

To illustrate this advantage, we conducted a simple numerical experiment. Details of the data generating mechanism are given in \ref{app:sim_gen_2}.
To ``estimate'' the nuisance functions $(\mu_a, \eta_a, q_s)$, we perturbed the true underlying functions as follows: $\widehat{\mu}_a(X) = \mu_a(X) + h \, \epsilon_{\mu}$, $\widehat{\eta}_{a}(X) = \mathrm{expit}\left\{\mathrm{logit}\left[\eta_{a}(X)\right] + 1.3h \, \epsilon_{\eta}\right\}$, and $\widehat{q}_{s}(X) = \mathrm{expit}\left\{\mathrm{logit}\left[q_{s}(X)\right] + 1.3 h \, \epsilon_{q}\right\}$, where $\epsilon_{\mu}, \epsilon_{\eta}, \epsilon_{q} \overset{\mathrm{iid}}{\sim} \mathcal{N}(n^{-r}, n^{-2r})$, setting $h = 2.5$. 
This construction guarantees that $\lVert \widehat{\mu}_a - \mu_a\rVert = O(n^{-r}), \lVert \widehat{\eta}_a - \eta_a\rVert = O(n^{-r}), \lVert \widehat{q}_s - q_s\rVert = O(n^{-r})$, so that we can study the performance of the proposed estimator with nuisance errors known to be $O(n^{-r})$.

For $n \in \{100, 500, 1000\}$, we computed nuisance estimates varying $r \in \{0.10 + 0.05k: k \in \{0, \ldots, 8\}\}$. For this simulation, the subgroup of interest was $\widetilde{X} \equiv X_0 = 1$, and the target parameter was $\phi_{1,1}(1)$. Aggregating over 5,000 replications for each scenario, we compared the root mean-square error (RMSE) of the proposed influence function-based estimator $\widehat{\phi}_{1,1}(1)$ to that of a plug-in estimator $\widehat \phi^{g}_{1,1}(1)$. Results are shown in Figure \ref{fig:rate}.
\begin{figure}
    \centering
    \includegraphics[width = 0.5\linewidth]{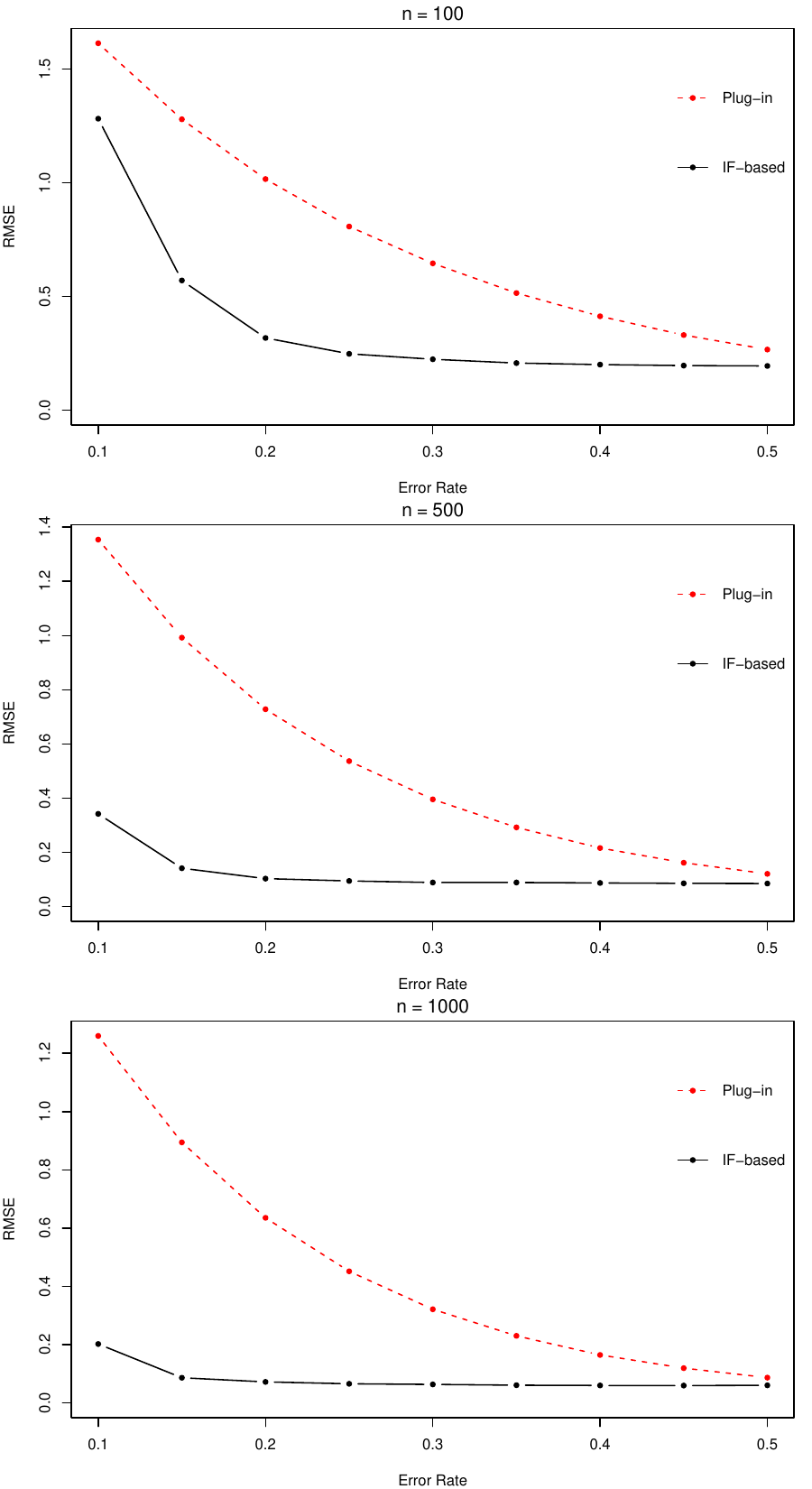}
    \caption{Simulation results for $\widehat \phi_{1,1}(1)$ for total sample size $\sum_{s=1}^{2}n_s$ 100, 500, and 1000 with 5000 iterations per scenario. Varying the error rate $r$ which controls the nuisance error at $O(n^{-r})$, the results include the root mean-square error (RMSE) of the proposed estimator (IF-based) and the estimator that only uses an outcome model (plug-in).}
    \label{fig:rate}
\end{figure}
According to the results in Figure~\ref{fig:rate}, the proposed doubly robust estimator outperforms the plug-in estimator, particularly when nuisance error is high (i.e., $r$ small). In fact, as predicted by Theorem~\ref{thm:estimation2}, near-optimal performance of the doubly robust estimator is achieved once $r \geq 0.25$. Such low RMSE is only achieved by the plug-in estimator when $r = 0.5$, which would only be anticipated when $\widehat{\mu}_a$ is a correctly specified parametric model. This situation is unlikely in practice.
Finally, note that similar results also apply to $\widehat \psi_{a}(\tilde x)$; the results for the analogous simulation study are omitted here.

\section{Application}\label{sec:Application}

Schizophrenia is a syndrome characterized by the presence of psychotic, negative, and cognitive  symptoms that typically emerges during the first half of adult life. The course and outcomes of schizophrenia vary markedly across individuals, and there is potential for recovery. However, in most cases, schizophrenia is a ``chronic, severe'' and ``highly disabling'' condition that drives substantial social and occupational dysfunction. \citep{schizophernia2023nimh} To rate the severity of schizophrenia symptoms, researchers frequently use the Positive and Negative Syndrome Scale (PANSS) total score (ranges from 30-200), and it closely relates to the clinical severity of patients \citep{leucht2005does}. The higher the score, the higher the severity of symptoms. Paliperidone extended-release (ER) is an oral antipsychotic agent  approved for treating schizophrenia by the US Food and Drug Administration and the European Agency for the Evaluation of Medicinal Products. 
\citet{emsley2008efficacy} evaluated efficacy and safety of paliperidone ER tablets (3 – 12 mg/day) for treating schizophrenia from three 6-week, placebo-controlled, double-blind trials
\citep{kane2007treatment,davidson2007efficacy,marder2007efficacy}. However, it is not clear how the effects vary according to the patient's characteristics and the effects in each target population (trial).

In this study, using the three trials, we estimate the difference in expected PANSS total score at week 6 after the treatment initiation between the assignment of paliperidone ER and placebo in different subgroups (stratified by gender, age, and PANSS baseline score) and target populations (represented by each trial) by the proposed method. 

We include the patients assigned to either paliperidone ER or placebo and whose PANSS total scores were evaluated at baseline and week 6. The summary statistics for each variable included in each trial are given in \ref{app:application1}. In total, 1218 patients were included. We include information on the patient's gender, age, race, and PANSS baseline score. The three trials represent different populations. For instance, regarding race, trial 1 was dominated by White, whereas trial 2 included mainly White and Black, and trial 3 mostly recruited White and Asians. The age distribution also presents heterogeneity among trials. We hypothesize that the treatment effects would vary across the population represented by different trials and the patient's PANSS total score at baseline, gender, and age group.
\begin{figure}
    \centering
    \includegraphics[width=0.3\textwidth]{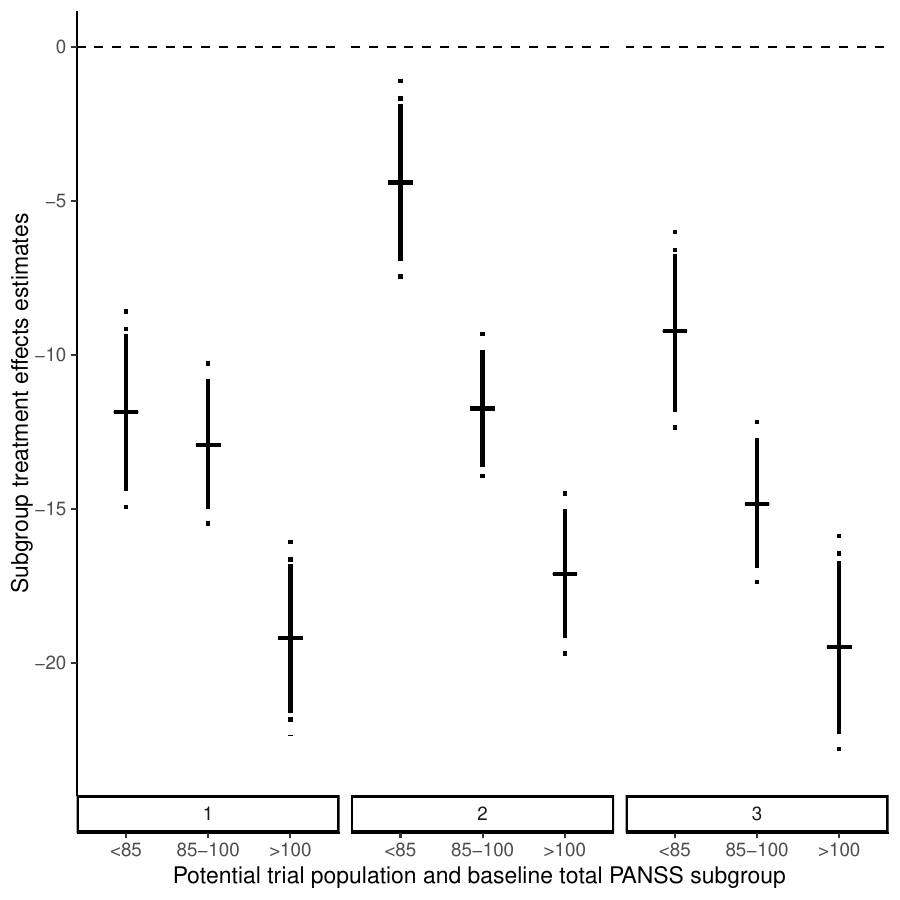}
    \caption{Results from analyses using the meta-trial data. For each trial (target population) and each PANSS total score baseline subgroup, the plot includes the point estimation (horizontal short solid line), 95\% confidence interval (vertical solid line), and simultaneous confidence bands (vertical dashed line). The horizontal dashed line indicates the null.}
        \label{fig:PANSSbaseline}
\end{figure}
We used GAM to estimate the outcome model and GLM for the treatment and source model. Figure \ref{fig:PANSSbaseline} shows the results where the subgroup is defined as PANSS total score baseline group ($<$85, 85-100, $>$100). Across all three target populations, patients with a higher PANSS baseline total score benefit most from the treatment but with different magnitudes. Target populations 2 and 3 exhibit similar heterogeneity but patients in target population 2 benefit more from the treatment. The difference of treatment effects between $<$85 and 85-100 subgroup are minimal in target population 1. This shows that the subgroup treatment effects can be different in different target populations, and the effects can be different in different subgroups. These stress the importance of investigating the subgroup effects rather than averaged effects and the effects in different target populations.

We additionally analyze the data to obtain the subgroup treatment effects stratified by the gender (with the collected data, we can only divide patients by female/male) and also the age group (18-34, 34-48, 48-74). The stratification is based on the first and third quantiles of the variable. The results are given in \ref{app:application2}. We can see that the treatment is more beneficial in females than males in trials 1 and 3. However, the treatment effects are similar in males and females in trial 2.  Paliperidone ER showed less effect in patients older than 48 in trials 1 and 3, but the patients who were 34-48 years old were less responsive to the treatment in trial 2. 
Both results show heterogeneity exists among subgroups and target populations.

\section{Discussion}

We have described methods to estimate a subgroup treatment effect for an internal or external target population using multi-source data. Our results can provide insights into one of the important problems in the field of ``evidence synthesis'': how to use the combined evidence to efficiently investigate the heterogeneity of treatment effects among a target population of substantive interest. Estimating subgroup treatment effects can answer how the effects vary according to a hypothesized categorical effect modifier, which can then be used to support causal explanations, facilitate personalized treatment recommendations, and generate hypotheses for future studies. 

Traditional techniques like meta-analysis for evidence synthesis \citep{higgins2009re,rice2018re} usually fall short in providing causally interpretable estimates from multi-source data. Conversely, applying conventional causal methodologies directly to multi-source data may lead to effect estimates that lack relevance to any specific, well-defined population of interest. In contrast, our causal method explicitly targets a well-defined population based on the interest of decision-making, which may be one of or different from the populations underlying the multi-source data, so that offers more relevant and actionable insights.


Our method provides a robust and efficient estimation of a treatment effect for a pre-specified subgroup among a target population. Leveraging the multi-source data, we can better explore the heterogeneity in a large sample, encompassing a large amount of heterogeneity. By employing semi-parametric theory and flexible machine learning methods, the developed estimator 1) is doubly robust, giving analysts two chances at consistency, 2) allows the use of flexible data-adaptive methods (possibly out of the Donsker class) to model the nuisance parameters with relatively slow convergence rates, and 3) is non-parametrically efficient. In addition, we take the multiplicity issue into account by providing simultaneous confidence bands.


We assume the transportability assumption for conditional difference measures (the mean difference of counterfactual outcomes is independent of the data source conditional on the covariates) to identify the target parameter. Such an assumption would need to be examined according to substantive knowledge. Other assumptions, such as transportability assumption for conditional relative effect measures \citep{dahabreh2022learning}, may also be able to identify the target parameters. We leave the work in the future.

\section*{Acknowledgments}
The authors thank Dr. Gonzalo Martínez-Alés for the insightful discussion regarding the data application.
This work was supported in part by Patient-Centered Outcomes Research Institute (PCORI) awards ME-2021C2-22365. The content is solely the responsibility of the authors and does not necessarily represent the official views of PCORI, PCORI's Board of Governors or PCORI's Methodology Committee. 
This study, carried out under The Yale University Open Data Access  (YODA) \url{https://yoda.yale.edu/} Project number 2022-5062, used data obtained from the Yale University Open Data Access Project, which has an agreement with Jassen Research \& Development, L.L.C.. The interpretation and reporting of research using this data are solely the responsibility of the authors and do not necessarily represent the official views of the Yale University Open Data Access Project or Jassen Research \& Development, L.L.C..
\vspace*{-8pt}

\section*{Data availability statement}

Code to reproduce our simulations and the data analyses are provided on GitHub: \href{https://github.com/Guanbo-W/Transportability_STE}{link}. \\
The data analyzed can be obtained from the YODA, subject to approval by the YODA \url{https://yoda.yale.edu/}

\section*{Supporting Information}
Web Appendices, Figures and Tables referenced in Sections 2 through 5, and a copy of the code uploaded on GitHub: \href{https://github.com/Guanbo-W/Transportability_STE}{link}.

\bibliographystyle{apalike}
\bibliography{biblio}

\begin{thebibliography}{}

\bibitem[Amatya et~al., 2021]{amatya2021subgroup}
Amatya, A.~K., Fiero, M.~H., Bloomquist, E.~W., Sinha, A.~K., Lemery, S.~J., Singh, H., Ibrahim, A., Donoghue, M., Fashoyin-Aje, L.~A., de~Claro, R.~A., et~al. (2021).
\newblock Subgroup analyses in oncology trials: regulatory considerations and case examples.
\newblock {\em Clinical Cancer Research}, 27(21):5753--5756.

\bibitem[Belloni et~al., 2015]{belloni2015some}
Belloni, A., Chernozhukov, V., Chetverikov, D., and Kato, K. (2015).
\newblock Some new asymptotic theory for least squares series: Pointwise and uniform results.
\newblock {\em Journal of Econometrics}, 186(2):345--366.

\bibitem[Bickel et~al., 1993]{bickel1993efficient}
Bickel, P.~J., Klaassen, C.~A., Bickel, P.~J., Ritov, Y., Klaassen, J., Wellner, J.~A., and Ritov, Y. (1993).
\newblock {\em Efficient and adaptive estimation for semiparametric models}, volume~4.
\newblock Springer.

\bibitem[Burke et~al., 2015]{burke2015three}
Burke, J.~F., Sussman, J.~B., Kent, D.~M., and Hayward, R.~A. (2015).
\newblock Three simple rules to ensure reasonably credible subgroup analyses.
\newblock {\em Bmj}, 351.

\bibitem[Chen et~al., 2022]{chen2022debiased}
Chen, Q., Syrgkanis, V., and Austern, M. (2022).
\newblock Debiased machine learning without sample-splitting for stable estimators.
\newblock {\em arXiv preprint arXiv:2206.01825}.

\bibitem[Chernozhukov et~al., 2018]{chernozhukov2018double}
Chernozhukov, V., Chetverikov, D., Demirer, M., Duflo, E., Hansen, C., Newey, W., and Robins, J. (2018).
\newblock Double/debiased machine learning for treatment and structural parameters.
\newblock {\em The Econometrics Journal}, 21(1):C1--C68.

\bibitem[Dahabreh et~al., 2019a]{dahabreh2019studydesigns}
Dahabreh, I.~J., Haneuse, S. J.-P., Robins, J.~M., Robertson, S.~E., Buchanan, A.~L., Stuart, E.~A., and Hern\'an, M.~A. (2019a).
\newblock Study designs for extending causal inferences from a randomized trial to a target population.
\newblock {\em arXiv preprint arXiv:1905.07764}.

\bibitem[Dahabreh et~al., 2019b]{dahabreh2019generalizing}
Dahabreh, I.~J., Hern{\'a}n, M.~A., Robertson, S.~E., Buchanan, A., and Steingrimsson, J.~A. (2019b).
\newblock Generalizing trial findings in nested trial designs with sub-sampling of non-randomized individuals.
\newblock {\em arXiv preprint arXiv:1902.06080}.

\bibitem[Dahabreh et~al., 2023]{dahabreh2023efficient}
Dahabreh, I.~J., Robertson, S.~E., Petito, L.~C., Hernán, M.~A., and Steingrimsson, J.~A. (2023).
\newblock Efficient and robust methods for causally interpretable meta-analysis: Transporting inferences from multiple randomized trials to a target population.
\newblock {\em Biometrics}, 79(2):1057--1072.

\bibitem[Dahabreh et~al., 2022]{dahabreh2022learning}
Dahabreh, I.~J., Robertson, S.~E., and Steingrimsson, J.~A. (2022).
\newblock Learning about treatment effects in a new target population under transportability assumptions for relative effect measures.
\newblock {\em arXiv preprint arXiv:2202.11622}.

\bibitem[Davidson et~al., 2007]{davidson2007efficacy}
Davidson, M., Emsley, R., Kramer, M., Ford, L., Pan, G., Lim, P., and Eerdekens, M. (2007).
\newblock Efficacy, safety and early response of paliperidone extended-release tablets (paliperidone er): results of a 6-week, randomized, placebo-controlled study.
\newblock {\em Schizophrenia research}, 93(1-3):117--130.

\bibitem[Emsley et~al., 2008]{emsley2008efficacy}
Emsley, R., Berwaerts, J., Eerdekens, M., Kramer, M., Lane, R., Lim, P., Hough, D., and Palumbo, J. (2008).
\newblock Efficacy and safety of oral paliperidone extended-release tablets in the treatment of acute schizophrenia: pooled data from three 52-week open-label studies.
\newblock {\em International Clinical Psychopharmacology}, 23(6):343--356.

\bibitem[{FDA}, 2018]{FDA2018}
{FDA} (2018).
\newblock {Developing Targeted Therapies in Low-Frequency Molecular Subsets of a Disease Guidance for Industry}.
\newblock \url{https://www.fda.gov/media/117173/download}.
\newblock Accessed: Jan. 3rd, 2024.

\bibitem[{FDA}, 2019]{FDA2019}
{FDA} (2019).
\newblock {Enrichment Strategies for Clinical Trials to Support Determination of Effectiveness of Human Drugs and Biological Products Guidance for Industry}.
\newblock \url{https://www.fda.gov/media/121320/download}.
\newblock Accessed: Jan. 3rd, 2024.

\bibitem[{FDA}, 2023]{FDA2023}
{FDA} (2023).
\newblock {Food and Drug Administration, Department of Health and Human Services, Code of Federal Regulation, 21 CFR 314.50}.
\newblock \url{https://www.ecfr.gov/current/title-21/part-314/section-314.50}.
\newblock Accessed: Jan. 3rd, 2024.

\bibitem[Higgins et~al., 2009]{higgins2009re}
Higgins, J.~P., Thompson, S.~G., and Spiegelhalter, D.~J. (2009).
\newblock A re-evaluation of random-effects meta-analysis.
\newblock {\em Journal of the Royal Statistical Society: Series A (Statistics in Society)}, 172(1):137--159.

\bibitem[Horowitz, 2009]{horowitz2009semiparametric}
Horowitz, J.~L. (2009).
\newblock {\em Semiparametric and nonparametric methods in econometrics}, volume~12.
\newblock Springer.

\bibitem[{ICH}, 1998]{ICH1998}
{ICH} (1998).
\newblock {European Medicines Agency, International Conference on Harmonization Topic E 9 Statistical Principles for Clinical Trials}.
\newblock \url{https://www.ema.europa.eu/en/documents/scientific-guideline/ich-e-9-statistical-principles-clinical-trials-step-5_en.pdf}.
\newblock Accessed: Jan. 3rd, 2024.

\bibitem[Jaman et~al., 2024]{jaman2024penalized}
Jaman, A., Wang, G., Ertefaie, A., Bally, M., Lévesque, R., Platt, R., and Schnitzer, M. (2024).
\newblock Penalized g-estimation for effect modifier selection in the structural nested mean models for repeated outcomes.

\bibitem[Kane et~al., 2007]{kane2007treatment}
Kane, J., Canas, F., Kramer, M., Ford, L., Gassmann-Mayer, C., Lim, P., and Eerdekens, M. (2007).
\newblock Treatment of schizophrenia with paliperidone extended-release tablets: a 6-week placebo-controlled trial.
\newblock {\em Schizophrenia Research}, 90(1-3):147--161.

\bibitem[Kang and Schafer, 2007]{kang2007demystifying}
Kang, J.~D. and Schafer, J.~L. (2007).
\newblock Demystifying double robustness: A comparison of alternative strategies for estimating a population mean from incomplete data.
\newblock {\em Statistical science}, 22(4):523--539.

\bibitem[Kennedy et~al., 2020]{kennedy2020sharp}
Kennedy, E.~H., Balakrishnan, S., and G’Sell, M. (2020).
\newblock Sharp instruments for classifying compliers and generalizing causal effects.
\newblock {\em The Annals of Statistics}, 48(4):2008--2030.

\bibitem[Kennedy et~al., 2019]{kennedy2019robust}
Kennedy, E.~H., Lorch, S., and Small, D.~S. (2019).
\newblock Robust causal inference with continuous instruments using the local instrumental variable curve.
\newblock {\em Journal of the Royal Statistical Society: Series B (Statistical Methodology)}, 81(1):121--143.

\bibitem[K{\"u}nzel et~al., 2019]{kunzel2019metalearners}
K{\"u}nzel, S.~R., Sekhon, J.~S., Bickel, P.~J., and Yu, B. (2019).
\newblock Metalearners for estimating heterogeneous treatment effects using machine learning.
\newblock {\em Proceedings of the national academy of sciences}, 116(10):4156--4165.

\bibitem[Leucht et~al., 2005]{leucht2005does}
Leucht, S., Kane, J.~M., Kissling, W., Hamann, J., Etschel, E., and Engel, R.~R. (2005).
\newblock What does the panss mean?
\newblock {\em Schizophrenia research}, 79(2-3):231--238.

\bibitem[Liu et~al., 2022]{liu2022modeling}
Liu, Y., Schnitzer, M.~E., Wang, G., Kennedy, E., Viiklepp, P., Vargas, M.~H., Sotgiu, G., Menzies, D., and Benedetti, A. (2022).
\newblock Modeling treatment effect modification in multidrug-resistant tuberculosis in an individual patient data meta-analysis.
\newblock {\em Statistical methods in medical research}, 31(4):689--705.

\bibitem[Marder et~al., 2007]{marder2007efficacy}
Marder, S.~R., Kramer, M., Ford, L., Eerdekens, E., Lim, P., Eerdekens, M., and Lowy, A. (2007).
\newblock Efficacy and safety of paliperidone extended-release tablets: results of a 6-week, randomized, placebo-controlled study.
\newblock {\em Biological Psychiatry}, 62(12):1363--1370.

\bibitem[{National Institute of Mental Health}, 2023]{schizophernia2023nimh}
{National Institute of Mental Health} (2023).
\newblock {Schizophrenia, U.S. Department of Health and Human Services, National Institutes of Health.}
\newblock Retrieved May 3, 2023, from \url{https://www.nimh.nih.gov/health/topics/schizophrenia}.

\bibitem[Oxman and Guyatt, 1992]{oxman1992consumer}
Oxman, A.~D. and Guyatt, G.~H. (1992).
\newblock A consumer's guide to subgroup analyses.
\newblock {\em Annals of internal medicine}, 116(1):78--84.

\bibitem[Pfanzagl, 2012]{pfanzagl2012}
Pfanzagl, J. (2012).
\newblock {\em Contributions to a general asymptotic statistical theory}, volume~13.
\newblock Springer Science \& Business Media.

\bibitem[Rice et~al., 2018]{rice2018re}
Rice, K., Higgins, J., and Lumley, T. (2018).
\newblock A re-evaluation of fixed effect (s) meta-analysis.
\newblock {\em Journal of the Royal Statistical Society: Series A (Statistics in Society)}, 181(1):205--227.

\bibitem[Robertson et~al., 2021]{robertson2021intercept}
Robertson, S.~E., Steingrimsson, J.~A., and Dahabreh, I.~J. (2021).
\newblock Using numerical methods to design simulations: revisiting the balancing intercept.
\newblock {\em American Journal of Epidemiology}.

\bibitem[Robins, 1988]{robins1988confidence}
Robins, J.~M. (1988).
\newblock Confidence intervals for causal parameters.
\newblock {\em Statistics in medicine}, 7(7):773--785.

\bibitem[Robins and Greenland, 2000]{robins2000d}
Robins, J.~M. and Greenland, S. (2000).
\newblock Causal inference without counterfactuals: comment.
\newblock {\em Journal of the American Statistical Association}, 95(450):431--435.

\bibitem[Robins et~al., 2000]{robins2000sensitivity}
Robins, J.~M., Rotnitzky, A., and Scharfstein, D.~O. (2000).
\newblock Sensitivity analysis for selection bias and unmeasured confounding in missing data and causal inference models.
\newblock In {\em Statistical models in epidemiology, the environment, and clinical trials}, pages 1--94. Springer.

\bibitem[Rothwell, 2005]{rothwell2005subgroup}
Rothwell, P.~M. (2005).
\newblock Subgroup analysis in randomised controlled trials: importance, indications, and interpretation.
\newblock {\em The Lancet}, 365(9454):176--186.

\bibitem[Rubin, 2010]{rubin2010reflections}
Rubin, D.~B. (2010).
\newblock Reflections stimulated by the comments of {S}hadish (2010) and {W}est and {T}hoemmes (2010).
\newblock {\em Psychological Methods}, 15(1):38--46.

\bibitem[Schick, 1986]{schick1986}
Schick, A. (1986).
\newblock On asymptotically efficient estimation in semiparametric models.
\newblock {\em The Annals of Statistics}, pages 1139--1151.

\bibitem[Schmid et~al., 2020]{schmid2020comparing}
Schmid, I., Rudolph, K.~E., Nguyen, T.~Q., Hong, H., Seamans, M.~J., Ackerman, B., and Stuart, E.~A. (2020).
\newblock Comparing the performance of statistical methods that generalize effect estimates from randomized controlled trials to much larger target populations.
\newblock {\em Communications in Statistics-Simulation and Computation}, pages 1--23.

\bibitem[Siddique et~al., 2019]{siddique2019causal}
Siddique, A.~A., Schnitzer, M.~E., Bahamyirou, A., Wang, G., Holtz, T.~H., Migliori, G.~B., Sotgiu, G., Gandhi, N.~R., Vargas, M.~H., Menzies, D., et~al. (2019).
\newblock Causal inference with multiple concurrent medications: A comparison of methods and an application in multidrug-resistant tuberculosis.
\newblock {\em Statistical methods in medical research}, 28(12):3534--3549.

\bibitem[Sun et~al., 2010]{sun2010subgroup}
Sun, X., Briel, M., Walter, S.~D., and Guyatt, G.~H. (2010).
\newblock Is a subgroup effect believable? updating criteria to evaluate the credibility of subgroup analyses.
\newblock {\em Bmj}, 340.

\bibitem[Sun et~al., 2014]{sun2014use}
Sun, X., Ioannidis, J.~P., Agoritsas, T., Alba, A.~C., and Guyatt, G. (2014).
\newblock How to use a subgroup analysis: users’ guide to the medical literature.
\newblock {\em Jama}, 311(4):405--411.

\bibitem[Vaart and Wellner, 1996]{vaart1996weak}
Vaart, A.~W. and Wellner, J.~A. (1996).
\newblock Weak convergence.
\newblock In {\em Weak convergence and empirical processes}, pages 16--28. Springer.

\bibitem[Van~der Vaart, 2000]{vandervaart2000asymptotic}
Van~der Vaart, A.~W. (2000).
\newblock {\em Asymptotic statistics}, volume~3.
\newblock Cambridge University Press.

\bibitem[VanderWeele, 2009]{vanderWeele2009}
VanderWeele, T.~J. (2009).
\newblock Concerning the consistency assumption in causal inference.
\newblock {\em Epidemiology}, 20(6):880--883.

\bibitem[Wang et~al., 2023]{wang2023evaluating}
Wang, G., Costello, M.~P., Pang, H., Zhu, J., Helms, H.-J., Reyes-Rivera, I., Platt, R.~W., Pang, M., and Koukounari, A. (2023).
\newblock Evaluating hybrid controls methodology in early-phase oncology trials: a simulation study based on the morpheus-uc trial.
\newblock {\em Pharmaceutical Statistics}.

\bibitem[Wang et~al., 2020]{wang2020estimating}
Wang, G., Schnitzer, M.~E., Menzies, D., Viiklepp, P., Holtz, T.~H., and Benedetti, A. (2020).
\newblock Estimating treatment importance in multidrug-resistant tuberculosis using targeted learning: An observational individual patient data network meta-analysis.
\newblock {\em Biometrics}, 76(3):1007--1016.

\bibitem[Wang et~al., 2007]{wang2007statistics}
Wang, R., Lagakos, S.~W., Ware, J.~H., Hunter, D.~J., and Drazen, J.~M. (2007).
\newblock Statistics in medicine—reporting of subgroup analyses in clinical trials.
\newblock {\em New England Journal of Medicine}, 357(21):2189--2194.

\end{thebibliography}





\clearpage
\textbf{\large\bf Supplementary material for ``Efficient estimation of subgroup treatment effects using multi-source data''}
\setcounter{page}{1}
\renewcommand{\thepage}{S\arabic{page}}
\begin{appendices}
\numberwithin{equation}{section}
\renewcommand{\theequation}{\thesection.\arabic{equation}}
\numberwithin{table}{section}
\renewcommand{\thetable}{\thesection.\arabic{table}}
\numberwithin{figure}{section}
\renewcommand{\thefigure}{\thesection.\arabic{figure}}
\numberwithin{assumption}{section}
\renewcommand{\theassumption}{\thesection.\arabic{assumption}}
\renewcommand{\thesection}{Web Appendix \Alph{section}}
\renewcommand{\thesubsection}{\Alph{section}.\arabic{subsection}}

\titleformat{\subsection}
      {\normalfont\fontsize{10}{12}\bfseries}{\thesubsection}{1em}{}
      

\renewcommand{\thetable}{S\arabic{table}}
\renewcommand{\thefigure}{S\arabic{figure}}
\renewcommand{\theequation}{S\arabic{equation}}
\renewcommand{\thetheorem}{S\arabic{theorem}}

\section{Identification of $\phi_{a, s}(\widetilde{x})$}
\label{appendix:identification2}
\begin{restatable}[Identification of subgroup potential outcome means in an internal data]{theorem}{thmidentificationcollection2}
\label{thm:identification2}
Under conditions \textit{A1} through \textit{A5}, the subgroup potential outcome mean in the internal source-specific data under treatment $a \in \mathcal A$, $\phi_{a, s}(\widetilde{x})=\E(Y^a | \widetilde{X}=\widetilde{x}, S=s )$, is identifiable by the observed data functional
\begin{equation} \label{eq:identification2g}
    \begin{split}
  \phi_{a, s} (\widetilde{x}) \equiv 
  \E\big\{\E(Y |A=a, X) | \widetilde{X}=\widetilde{x}, S=s\big\},
    \end{split}
\end{equation}
which can be equivalently expressed as 
\begin{equation} \label{eq:identification2w}
    \begin{split}
  \dfrac{1}{\Pr(\widetilde{X}=\widetilde{x}, S=s)}\E\Big\{\dfrac{I(A=a, \widetilde{X}=\widetilde{x})Y \Pr(S=s|X)}{\Pr(A=a|X)}\Big\}.
    \end{split}
\end{equation}
\end{restatable}
\begin{proof}
By the law of total expectation and Assumptions \textit{A2, A4, A1} respectively, we have
\begin{align*}
    \phi_{a, s}(\widetilde{x})= 
     & \E(Y^{a} | \widetilde{X}=\widetilde{x}, S=s)\\ 
    =& \E\big\{\E(Y^{a} | X, S=s) | \widetilde{X}=\widetilde{x}, S=s\big\}\\ 
    =& \E\big\{\E(Y^{a} | A=a, X, S=s) | \widetilde{X}=\widetilde{x}, S=s\big\}\\  
    =& \E\big\{\E(Y^{a} | A=a, X) | \widetilde{X}=\widetilde{x}, S=s\big\}\\ 
    =& \E\big\{\E(Y | A=a, X) | \widetilde{X}=\widetilde{x}, S=s\big\},
\end{align*}
which completes the derivation of the result in \eqref{eq:identification2g}. We can re-express \eqref{eq:identification2g} to use weighting (which relays on Assumptions \textit{A3} and \textit{A5}),
\begin{align*}
    \phi_{a, s}(\widetilde{x})
    =& \E\big\{\E(Y | A=a, X) | \widetilde{X}=\widetilde{x}, S=s\big\}\\
    =& \E\Bigg[\E\Big\{\dfrac{I(A=a)Y}{\Pr(A=a|X)}\Big| X\Big\}\Big|\widetilde{X}=\widetilde{x}, S=s\Bigg]\\
    =& \dfrac{1}{\Pr(\widetilde{X}=\widetilde{x}, S=s)}\E\Bigg[I(\widetilde{X}=\widetilde{x}, S=s)E\Big\{\dfrac{I(A=a )Y}{\Pr(A=a|X)}\Big| X\Big\}\Bigg]\\
    =& \dfrac{1}{\Pr(\widetilde{X}=\widetilde{x}, S=s)}\E\Big\{\dfrac{I(A=a, \widetilde{X}=\widetilde{x})Y \Pr(S=s|X)}{\Pr(A=a|X)}\Big\}.
\end{align*}
\end{proof}
\section{The first-order influence function of $\phi_{a, s}(\widetilde{x})$ under the non-parametric model}
\label{appendix:influence_function2}
The first-order influence function of $\phi_{a, s}(\widetilde{x})$ under the non-parametric model is 
\begin{align*}
\mathit\Phi_ {p_{0}}(a, s, \widetilde{x})=&\dfrac{1}{\Pr_ {p_{0}}(\widetilde{X}=\widetilde{x}, S=s)}
    \Big[I(\widetilde{X}=\widetilde{x}, S=s)\big\{\E_{p_{0}}(Y | A=a, X)-\phi_{a, s}(\widetilde{x})\big\}\\
    +&I(A=a, \widetilde{X}=\widetilde{x})\dfrac{\Pr_{p_{0}}(S=s|X)}{\Pr_{p_{0}}(A=a|X)}\big\{Y-\E_{p_{0}}(Y | A=a, X)\big\}\Big].
\end{align*}
\begin{proof}
Recall that under the nonparametric model, $\mathcal M_{\text{\tiny np}}$, for the observable data, $O = (X,S,A,Y)$, the density of the law of the observable data can be written as 
\begin{equation*}
    p(s,x,a,y) = p(x) p(s|x) p(a|x, s) p(y|a, x, s).
\end{equation*}
Under this model, the tangent space is the Hilbert space of mean zero random variables with finite variance. It can be decomposed as $L_2^0 = \Lambda_{X} \oplus \Lambda_{S|X} \oplus  \Lambda_{A|X, S} \oplus  \Lambda_{Y|A, X, S} $. We now derive the influence function for $\phi_{a, s}(\widetilde{x})$ under this nonparametric model. 

\noindent
Recall that
\begin{equation*} 
  \phi_{a, s}(\widetilde{x}) \equiv \E\big\{\E(Y | A=a, X) | \widetilde{X}=\widetilde{x}, S=s\big\}.
\end{equation*}
We will use the path differentiability of $\phi_{a, s}(\widetilde{x})$ to obtain the efficient influence function under the non-parametric model for the observed data \citep{bickel1993efficient}. To do so, we examine the derivative of $\phi_{a, s, p_t}(\widetilde{x})$ with respect to $t$; where the subscript $p_t$ denotes the dependence of $\psi_{a, s}(\widetilde{x})$ on a one-dimensional parametric sub-model $p_t$, indexed by $t \in [0,1)$, with $t = 0$ denoting the ``true'' data law. Furthermore, denote the function $u_{p}(\cdot)$ as the score function of the observed data with the (conditional) likelihood $p$.\\
Utilizing the law of total expectation, the tricks of $\E\big\{\E(u|w)v\big\}=\E\big\{\E(v|w)u\big\}$, subtracting the mean zero terms, and $\E(A|B)=\E\Big\{\dfrac{I(B=b)}{\Pr(B=b)}A\Big\}$ when $B$ is discrete repeatedly, we have
\begin{align*}
     &\dfrac{\partial \phi_{p_{t}}(a, \widetilde{x})}{\partial t}\Big|_{t=0}\\
    =&\dfrac{\partial}{\partial t}\E_{p_{t}}\big\{\E_{p_{t}}(Y | A=a, X) | \widetilde{X}=\widetilde{x}, S=s\big\}\Big|_{t=0}\\
    =&\dfrac{\partial}{\partial t}\E_{p_{t}}\big\{\E_{p_{0}}(Y | A=a, X) | \widetilde{X}=\widetilde{x}, S=s\big\}\Big|_{t=0}\\
    +&\E_{p_{0}}\Big\{\dfrac{\partial}{\partial t}\E_{p_{t}}(Y | A=a, X)\Big|_{t=0} \widetilde{X}=\widetilde{x}, S=s\Big\}\\
    =&\E_{p_{0}}\big\{E_{p_{0}}(Y | A=a, X) u_{Y|A, X} | \widetilde{X}=\widetilde{x}, S=s\big\}\\
    +& \E_{p_{0}}\big\{\E_{p_{0}}(Y u_{Y|A=a, X}| A=a, X) |\widetilde{X}=\widetilde{x}, S=s\big\}\\
    =&\E_{p_{0}}\Big\{\dfrac{I(\widetilde{X}=\widetilde{x}, S=s)}{\Pr_ {p_{0}}(\widetilde{X}=\widetilde{x}, S=s)}\E_{p_{0}}(Y | A=a, X) u_{Y|A, X}  \Big\}\\
    +& \E_{p_{0}}\Big(\E_{p_{0}}\Big[\big\{Y-\E_{p_{0}}(Y | A=a, X)\big\}u_{Y|A=a, X}| A=a, X\Big]\Big|\widetilde{X}=\widetilde{x}, S=s\Big)\\
    =&\E_{p_{0}}\Big[\dfrac{I(\widetilde{X}=\widetilde{x}, S=s)}{\Pr_ {p_{0}}(\widetilde{X}=\widetilde{x}, S=s)}\big\{\E_{p_{0}}(Y | A=a, X  )-\phi_{a, s}(\widetilde{x})\big\} u_{Y|A, X}  \Big]\\
    +& \E_{p_{0}}\Big(\dfrac{I(\widetilde{X}=\widetilde{x}, S=s)}{\Pr_ {p_{0}}(\widetilde{X}=\widetilde{x}, S=s)}\E_{p_{0}}\Big[\big\{Y-\E_{p_{0}}(Y | A=a, X)\big\}u(O)| A=a, X  \Big]\Big)\\
    =&\E_{p_{0}}\Big[\dfrac{I(\widetilde{X}=\widetilde{x}, S=s)}{\Pr_ {p_{0}}(\widetilde{X}=\widetilde{x}, S=s)}\big\{\E_{p_{0}}(Y | A=a, X)-\phi_{a, s}(\widetilde{x})\big\} u(O)  \Big]\\
    +& \E_{p_{0}}\Big[\dfrac{I(A=a, \widetilde{X}=\widetilde{x})\Pr_{p_{0}}(S=s|X)}{\Pr_ {p_{0}}(\widetilde{X}=\widetilde{x}, S=s)\Pr_{p_{0}}(A=a|X)}\big\{Y-\E_{p_{0}}(Y | A=a, X)\big\}u(O)\Big]\\
    =&\E_{p_{0}}\Big(u(O)\dfrac{I(\widetilde{X}=\widetilde{x}, S=s)}{\Pr_ {p_{0}}(\widetilde{X}=\widetilde{x}, S=s)}
    \Big[\big\{\E_{p_{0}}(Y | A=a, X)-\phi_{a, s}(\widetilde{x})\big\}\\
    +&\dfrac{I(A=a, \widetilde{X}=\widetilde{x})\Pr_{p_{0}}(S=s|X)}{\Pr_ {p_{0}}(\widetilde{X}=\widetilde{x}, S=s)\Pr_{p_{0}}(A=a|X)}\big\{Y-\E_{p_{0}}(Y | A=a, X)\big\}\Big]
    \Big).
\end{align*}
Thus, we conclude that the influence function of $\phi_{a, s}(\widetilde{x})$ is $\mathit\Phi_{p_0}(a, \widetilde{x})$.
\end{proof}
\section{Proof of Corollary \ref{cor:identification2}}
\label{appendix:corollary}
\begin{restatable}[Efficient influence function under semiparametric model]{corollary}{cor:identification2}
\label{cor:identification2}
The efficient influence function under the semi-parametric model which incorporating the constrain $Y \indep S | (X, A=a)$, or/and knowing the propensity score $\Pr(A|X, S=s)$ (e.g., the multi-source data are a collection of randomized clinical trials) is $\mathit\Phi_ {p_{0}}(a, s, \widetilde{x})$.
\end{restatable}
\begin{proof}
Under the non-parametric model, the Hilbert space can be decomposed as $L_2^0 = \Lambda_{S} \oplus \Lambda_{X|S} \oplus  \Lambda_{A|X, S} \oplus  \Lambda_{Y|A, X, S}$. Now after incorporating the constrain $Y \indep S | (X, A=a)$ and knowing the propensity score $\Pr(A|X, S=s)$ will change the decomposition to $L_2^0 = \Lambda_{S} \oplus \Lambda_{X|S} \oplus  \Lambda_{Y|A, X}=\Lambda_{S} \oplus \Lambda_{W|\widetilde{X}, S} \oplus \Lambda_{\widetilde{X}|S} \oplus  \Lambda_{Y|A, X}$, where $X=(W, \widetilde{X})$.\\
Rewrite $\mathit\Phi_ {p_{0}}(a, s, \widetilde{x})$ as 
\begin{align*}
\mathit\Phi_ {p_{0}}(a, s, \widetilde{x})=&\dfrac{I(\widetilde{X}=\widetilde{x}, S=s)}{\Pr_ {p_{0}}(\widetilde{X}=\widetilde{x}, S=s)}
    \big\{\E_{p_{0}}(Y | A=a, X)-\phi_{a, s}(\widetilde{x})\big\}+\\
    &\dfrac{I(\widetilde{X}=\widetilde{x}, A=a)}{\Pr_ {p_{0}}(\widetilde{X}=\widetilde{x}, S=s)}
    \dfrac{\Pr_{p_{0}}(S=s|X)}{\Pr_{p_{0}}(A=a|X)}\big\{Y-\E_{p_{0}}(Y | A=a, X)\big\}.
\end{align*}
According to the total law of expectation, the first term is a function of $(X, S)$ and is mean zero conditional on $(\widetilde{X}, S)$. Thus it belongs to $\Lambda_{W|\widetilde{X}, S}$. The second term is a function of $(X, A, Y)$ and is mean zero conditional on $(X, A)$. Thus it belongs to $\Lambda_{Y|X, A}$. From these observations, we conclude that $\mathit\Phi_ {p_{0}}(a, s, \widetilde{x})$ belongs to $\Lambda_{W|\widetilde{X}, S}\oplus \Lambda_{Y|X, A}$, which indicates that the unique influence function $\mathit\Phi_ {p_{0}}(a, s, \widetilde{x})$ under the non-parametric model is also the efficient influence function under the semiparametric model which incorporate the two restrictions.\\
It is also trivial to see that $\mathit\Phi_ {p_{0}}(a, s, \widetilde{x})$ is also the efficient influence function under the semiparametric model which incorporate either constrain.
\end{proof}
\section{Proof of Lemma ~\ref{lemma:gamma}}
\label{app:lemma1}
\begin{restatable}{lemma}{lemmagamma}
\label{lemma:gamma}
For the above defined $\widehat \kappa(\widetilde{x}) = \{n^{-1} \sum_{i=1}^n I(\widetilde{X}=\widetilde{x}, S_i = s)\}^{-1}$,
\begin{enumerate}
    \item[(i)] $I(\widehat \kappa(\widetilde{x})^{-1}=0)=o_p(n^{-1/2})$
    \item[(ii)] $\widehat \kappa(\widetilde{x})\xrightarrow{P} \kappa(\widetilde{x})$.
\end{enumerate}
\end{restatable}
\begin{proof}
For any $\varepsilon>0$,
\begin{align*}
    &\Pr\{|\sqrt{n}I(\widehat\kappa(\widetilde{x})^{-1}=0)-0|>\varepsilon\}\\
    \leqslant 
    & \Pr\{\widehat\kappa(\widetilde{x})^{-1}=0\}=1-\{\kappa(\widetilde{x})^{-1}\}^{n}\\
    \rightarrow 
    & \text{ } 0, \quad \mathrm{ as } \quad n\rightarrow \infty \quad \textrm{(By assumption \textit{A5})}.
\end{align*}
By Weak Law of Large Numbers, $\forall \varepsilon$, we have
$
    \Pr\{|\widehat\kappa(\widetilde{x})^{-1}-\kappa(\widetilde{x})^{-1}|>\varepsilon\},
$
as $n\rightarrow\infty$.
Next, fix any $\varepsilon>0$, set $\kappa\equiv\Pr\{\widetilde{X}=\widetilde{x}, S=s\}/2$, we have
\begin{align*}
    &\Pr\{|\widehat\kappa(\widetilde{x})-\kappa(\widetilde{x})|>\varepsilon\}\\
    \leqslant
    & \Pr\Big\{\widehat\kappa(\widetilde{x})^{-1}\geqslant\kappa, \Big|\widehat\kappa(\widetilde{x})-\kappa(\widetilde{x})\Big|>\varepsilon\Big\}+\Pr\{\widehat\kappa(\widetilde{x})^{-1}<\kappa\}\\
    \leqslant
    & \Pr\{|\widehat\kappa(\widetilde{x})^{-1}-\kappa(\widetilde{x})^{-1}|\geqslant2\kappa^{2}\varepsilon\}+\Pr\{|\widehat\kappa(\widetilde{x})^{-1}-\kappa(\widetilde{x})^{-1}|>\kappa\}\\
    \rightarrow
    & \text{ } 0, \quad \mathrm{ as } \quad n\rightarrow \infty
\end{align*}
\end{proof}
\section{Proof of Theorem~\ref{thm:estimation2}}
\label{app:Proof-Estimation2}
\begin{restatable}{theorem}{thmestimationSA}
\label{thm:estimation1}
If assumptions \textit{A1} through \textit{A6}, and \textit{(b1)} through \textit{(b3)} hold, then 
\begin{align*}
     \widehat \psi_a(\widetilde{x})-\psi_a(\widetilde{x})
    =\mathbb{P}_n\{\mathit\Psi_ {p_{0}}(a, \widetilde{x})\}+
    Q_{n}+o_p(n^{-1/2}),
\end{align*}
where $Q_{n}\lesssim O_p\big\{||\widehat g_a(X) -g_a(X)||(||\widehat e_a(X) -e_a(X)||+||\widehat p_a(X) -p_a(X)||)\big\}$, and nuisance parameters are estimated on a separate independent sample.\\
In particular, if $Q_{n} = o_p(n^{-1/2})$, then 
$
    \sqrt{n}\big\{\widehat \psi_a(\widetilde{x})-\psi_a(\widetilde{x})\big\}\rightarrow\mathcal{N}\Big[0, \E_{p_{0}}\big\{\mathit\Psi_ {p_{0}}(a, \widetilde{x})^{2}\big\}\Big].
$\\
That is, $\widehat \psi_a(\widetilde{x})$ is non-parametric efficient.
\end{restatable}
\begin{proof}
Let the nuisance parameter $\boldsymbol{\lambda}\equiv\{\kappa(\widetilde{x}), \mu_{a}(X), \eta_a(X), q_{s}(X)\}$ denote the asymptotic limits (assumed to exist) of $\boldsymbol{\widehat\lambda}\equiv\{\widehat\kappa(\widetilde{x}),  \widehat \mu_{a}(X), \widehat \eta_a(X), \widehat q_{s}(X)$\}. \\
For general functions $\kappa', \mu'_{a}(X), \eta'_{a} (X),$ and $q'(X)$ define 
\begin{align*}
G(\boldsymbol{\lambda}') =   \kappa'(\widetilde{x}) \Bigg[ I(\widetilde{X}=\widetilde{x}, S = s) \mu'_{a}(X) +
&I(A =a, \widetilde{X}=\widetilde{x}) \dfrac{ q'(X) }{ \eta'_{a} (X)} \Big\{Y  -  \mu'_{a}(X)  \Big\}    \Bigg].
\end{align*}
Using notation from \citet{vaart1996weak}, define $\mathbb{P}_n\big\{v(W)\big\} = n^{-1} \sum_{i=1}^n v(W_i)$ for some function $v$ and a random variable $W$. Using this notation, $\widehat \phi_{a, s}(\widetilde{x}) = \mathbb{P}_n\big\{G(\boldsymbol{\widehat\lambda})\big\}$.\\
We observe that we can decompose $\widehat \phi_{a, s}(\widetilde{x}) - \phi_{a, s}(\widetilde{x})$ into three parts as below.
\begin{align*}
     &\widehat \phi_{a, s}(\widetilde{x}) - \phi_{a, s}(\widetilde{x})=\mathbb{P}_n\big\{G(\boldsymbol{\widehat \lambda})\big\}-\mathbb{P}\big\{G(\boldsymbol{\lambda})\big\}\\
    =&\underbrace{
    (\mathbb{P}_n-\mathbb{P})\big\{G(\boldsymbol{\widehat \lambda})-G(\boldsymbol{\lambda})\big\} }_{1}+
    \underbrace{
    \mathbb{P}\big\{G(\boldsymbol{\widehat \lambda})-G(\boldsymbol{\lambda})\big\} }_{2}+
    \underbrace{
    (\mathbb{P}_n-\mathbb{P})G(\boldsymbol{\lambda})}_{3}.
\end{align*}
Working on term 1, note that
\begin{align*}
    &||G(\boldsymbol{\widehat \lambda})-G(\boldsymbol{\lambda})||\\
    =
    & \Bigg|\Bigg|I(\widetilde{X}=\widetilde{x})\widehat \kappa(\widetilde{x})\Bigg[I(S=s)\widehat \mu_a(X)+I(A=a)\dfrac{\widehat q(X)}{\widehat \eta_a(X)}\Big\{Y-\widehat \mu_a(X) \Big\}  \Bigg]\\
    & -I(\widetilde{X}=\widetilde{x}) \kappa(\widetilde{x})\Bigg[I(S=s) \mu_a(X)+I(A=a)\dfrac{q(X)}{\eta_a(X)}\Big\{Y- \mu_a(X) \Big\}  \Bigg]\Bigg|\Bigg|\\
    \leqslant 
    & \Bigg|\Bigg|I(\widetilde{X}=\widetilde{x})\widehat \kappa(\widetilde{x})\Bigg[I(S=s)\{\widehat \mu_a(X)- \mu_a(X)\}+I(A=a)\dfrac{\widehat q(X)}{\widehat \eta_a(X)}\Big\{ \mu_a(X)- \widehat \mu_a(X) \Big\}  \Bigg]\Bigg|\Bigg|\\
    & + \Bigg|\Bigg|I(\widetilde{X}=\widetilde{x}) \{\widehat \kappa(\widetilde{x})-\kappa(\widetilde{x})\}\Bigg[I(S=s) \mu_a(X)+I(A=a)\dfrac{q(X)}{\eta_a(X)}\Big\{Y- \mu_a(X) \Big\}  \Bigg]\Bigg|\Bigg|\\
    & + \Bigg|\Bigg|I(A=a, \widetilde{X}=\widetilde{x})\widehat \kappa(\widetilde{x}) \Big\{\dfrac{\widehat q(X)}{\widehat \eta_a(X)}-\dfrac{q(X)}{\eta_a(X)}\Big\}\Big\{ Y- \mu_a(X)\Big\}\Bigg|\Bigg|\\
    \lesssim
    & || \widehat \mu_a(X)-\mu_a(X)||+||\widehat \eta_a(X)-\eta_a(X)||+||\widehat q(X)-q(X)|| +||\widehat \kappa(\widetilde{x})-\kappa(\widetilde{x})||
\end{align*}
By assumptions \textit{(a3)} and Lemma \ref{lemma:gamma}, we have $||G(\boldsymbol{\widehat \lambda})-G(\boldsymbol{\lambda})|| = o_p(1)$, so that by Lemma 2 of \citet{kennedy2020sharp}, $(\mathbb{P}_n-\mathbb{P})\big\{G(\boldsymbol{\widehat \lambda})-G(\boldsymbol{\lambda})\big\}=o_p(n^{-1/2})$.

Working on term 2, we have
\begin{align*}
     &\E\{G(\boldsymbol{\widehat \lambda})-G(\boldsymbol{\lambda})\}\\
    =& \E\Bigg(I(\widetilde{X}=\widetilde{x})
    \Big\{
    \widehat \kappa(\widetilde{x})\bigg[I(S=s)\widehat \mu_a(X)+I(A=a)\dfrac{\widehat q(X)}{\widehat \eta_a(X)}\Big\{ Y- \widehat \mu_a(X) \Big\} 
    \bigg]\\
    &-\kappa(\widetilde{x})I(S=s) \mu_a(X)
    \Big\}\Bigg) \\
    =&\E\Bigg\{I(\widetilde{X}=\widetilde{x})
    \Big(\widehat \kappa(\widetilde{x})\Big[q(X)\{\widehat \mu_a(X) - \mu_a(X) \}
    +\widehat q(X)\dfrac{\eta_a(X)}{\widehat \eta_a(X)}\Big\{\widehat \mu_a(X) - \mu_a(X) \Big\}\Big]
    \Big)\Bigg\}\\
    & + \{\widehat \kappa(\widetilde{x})- \kappa(\widetilde{x})\}\E\big[q(X)\mu_a(X)\big]\\
    =& \E\Bigg\{I(\widetilde{X}=\widetilde{x})
    \Big(\widehat \kappa(\widetilde{x})\Big[\{\widehat q(X)- q(X)\}\{\widehat \mu_a(X) - \mu_a(X) \}\\
    & +\widehat q(X) \Big\{\dfrac{\eta_a(X)}{\widehat \eta_a(X)}-1\Big\}\Big\{\widehat \mu_a(X) - \mu_a(X) \Big\}\Big]
    \Big)\Bigg\}+ \Big\{\dfrac{\widehat\kappa(\widetilde{x})}{\kappa(\widetilde{x})} -1\Big\}\phi_{a, s}(\widetilde{x}).
\end{align*}
Combine term 2 and 3, we have
\begin{align}\nonumber
     &\mathbb{P}\big\{G(\boldsymbol{\widehat \lambda})-G(\boldsymbol{\lambda})\big\}+
    (\mathbb{P}_n-\mathbb{P})G(\boldsymbol{\lambda})\\\label{6}
    =& \mathbb{P}_n\{\mathit\Phi_ {p_{0}}(a, s, \widetilde{x})\}\\\nonumber
    &+\E\Bigg\{I(\widetilde{X}=\widetilde{x})
    \Big(\widehat \kappa(\widetilde{x})\Big[\{\widehat q(X)- q(X)\}\{\widehat \mu_a(X) - \mu_a(X) \}\\\label{7}
    & +\widehat q(X) \Big\{\dfrac{\eta_a(X)}{\widehat \eta_a(X)}-1\Big\}\Big\{\widehat \mu_a(X) - \mu_a(X) \Big\}\Big]
    \Big)\Bigg\}\\
    &+\Big\{\dfrac{\widehat\kappa(\widetilde{x})^{-1}}{\kappa(\widetilde{x})^{-1}} -1\Big\}\phi_{a, s}(\widetilde{x})+\Big\{\dfrac{\widehat\kappa(\widetilde{x})}{\kappa(\widetilde{x})} -1\Big\}\phi_{a, s}(\widetilde{x})\label{8}
\end{align}
Factoring and simplifying term (\ref{8}) yields $\phi_{a, s}(\widetilde{x})\left\{\widehat\kappa(\widetilde{x}) - \kappa(\widetilde{x})\right\}\left\{\frac{1}{\kappa(\widetilde{x})} - \frac{1}{\widehat\kappa(\widetilde{x})}\right\}$, which is \\
$o_{p}(n^{-1/2})$ by the central limit theorem. That is, both term 1 and (\ref{8}) are $o_{p}(n^{-1/2})$.

Therefore, combine term 1, 2 and 3, we conclude
\begin{align*}
     \widehat \phi_{a, s}(\widetilde{x})-\phi_{a, s}(\widetilde{x})
    =&
    O_{p}\Big[||\widehat \mu_a(X) -\mu_a(X)||\big\{||\widehat \eta_a(X) -\eta_a(X)||+||\widehat q(X) -q(X)||\big\}\Big]+\\
    &\mathbb{P}_n\{\mathit\Phi_ {p_{0}}(a, s, \widetilde{x})\}+o_p(n^{-1/2}).
\end{align*}
Further, if $\left\lVert \widehat\mu_a - \mu_a\right\rVert \left\{\left\lVert \widehat\eta_a - \eta_a\right\rVert + \left\lVert \widehat q_a - q_a\right\rVert\right\} = o_p(n^{-1/2})$, then
\begin{align*}
    \sqrt{n}\big\{\widehat \phi_{a, s}(\widetilde{x})-\phi_{a, s}(\widetilde{x})\big\}\rightarrow\mathcal{N}\Big[0, \E_{p_{0}}\big\{\mathit\Phi_ {p_{0}}(a, s, \widetilde{x})^{}\big\}\Big].
\end{align*}
That is, $\widehat \phi_{a, s}(\widetilde{x})$ is non-parametric efficient.
\end{proof}
\section{Proof of Corollary \ref{lemma:simultaneous confidence bands_bootstrap}}
\label{appendix:simultaneousconfidencebands_bootstrap}
\begin{restatable}{corollary}{}
\label{lemma:simultaneous confidence bands_bootstrap}
Let $\widehat \phi_{a, s}(\widetilde{x})$ and $\widehat\sigma_{a,s}(\widetilde{x})$ be the estimated $\phi_{a, s}(\widetilde{x})$ and its standard deviation respectively.
Suppose $\widetilde{x}$ is a vector with support $\widetilde{X}$, and let $$
t^b_{max}=\text{sup}_{\widetilde{x}\in\widetilde{X}}\Big|\dfrac{\widehat \phi_{a, s}^b(\widetilde{x})-\widehat \phi_{a, s}(\widetilde{x})}{\widehat\sigma_{a,s}(\widetilde{x})}\Big|,
$$
where $\widehat \phi_{a, s}^b(\widetilde{x})$ is the point estimate of $\widehat \phi_{a, s}(\widetilde{x}), b=1,\dots,B$ from the $b$-th bootstrapped data. If assumptions (a1) through (a3) hold, then the $(1-\alpha)$ simultaneous confidence bands of $\widehat \phi_{a, s}(\widetilde{x})$ is $\{\widehat \phi_{a, s}(\widetilde{x})-c(1-\alpha)\widehat \sigma_{a,s}(\widetilde{x}),\text{ } \widehat \phi_{a, s}(\widetilde{x})+c(1-\alpha)\widehat \sigma_{a,s}(\widetilde{x})\}$, where the critical value $c(1-\alpha)$ is the $(1-\alpha)$-quantile empirical distribution of $t^b_{max}$ over $B$ replicates.
\end{restatable}
\begin{proof}
The simultaneous confidence band with coverage probability $(1-\alpha)$ satisfies the following condition:
\begin{align*}
    &P\big\{\widehat \phi_{a, s}(\widetilde{x})-c(1-\alpha)\widehat \sigma_{a,s}(\widetilde{x}) \leqslant \phi_{a, s}(\widetilde{x}) \leqslant \widehat \phi_{a, s}(\widetilde{x})+c(1-\alpha)\widehat \sigma_{a,s}(\widetilde{x}), \text{ }\forall \widetilde{x}\in\widetilde{\mathcal{X}}\big\}=1-\alpha\\
    \implies &P\big\{\text{sup}_{\widetilde{x}\in\widetilde{\mathcal{X}}}\Big|\frac{\widehat \phi_{a, s}(\widetilde{x})-\phi_{a, s}(\widetilde{x})}{\widehat \sigma_{a,s}(\widetilde{x})}\Big|\leqslant c(1-\alpha)\big\}=1-\alpha
\end{align*}
Therefore, when setting $c(1-\alpha)$ as the $(1-\alpha)$-quantile empirical distribution of $t^b_{max}$ over $B$ replicates (Suppose B is large enough), the $(1-\alpha)$ simultaneous confidence bands of $\widehat \phi_{a, s}(\widetilde{x})$ is $\{\widehat \phi_{a, s}(\widetilde{x})-c(1-\alpha)\widehat \sigma_{a,s}(\widetilde{x}),\text{ } \widehat \phi_{a, s}(\widetilde{x})+c(1-\alpha)\widehat \sigma_{a,s}(\widetilde{x})\}$.
\end{proof}

\section{Approximation to the simultaneous confidence bands}
\label{app:SCB_another}
\begin{restatable}{corollary}{}
\label{lemma:simultaneous confidence bands}
Let $\widehat \phi_{a, s}(\widetilde{x})$ and $\widehat\sigma_{a,s}(\widetilde{x})$ be the estimated $\phi_{a, s}(\widetilde{x})$ and its standard deviation respectively. Suppose $\widetilde{x}$ is a vector with $d$ levels, and denote $\mathcal{N}_{d}^{b}$ by the $b$-th bootstrap draw from $N(0, I_d), b=1...B$. If $\left\lVert \widehat\mu_a(X) - \mu_a(X)\right\rVert \left\{\left\lVert \widehat\eta_a(X) - \eta_a(X)\right\rVert + \left\lVert \widehat q_s(X) - q_s(X)\right\rVert\right\} = o_p(n^{-1/2})$ and assumptions (a1) through (a3) hold, then the $(1-\alpha)$ asymptotic simultaneous confidence bands of $\widehat \phi_{a, s}(\widetilde{x})$ is $\{\widehat \phi_{a, s}(\widetilde{x})-c(1-\alpha)\widehat \sigma_{a,s}(\widetilde{x}),\text{ } \widehat \phi_{a, s}(\widetilde{x})+c(1-\alpha)\widehat \sigma_{a,s}(\widetilde{x})\}$, where the critical value $c(1-\alpha)$ is the $(1-\alpha)$-quantile empirical distribution of $\text{max}(|\mathcal{N}_{d}^{b}|)$ over $B$ replicates.
\end{restatable}
\noindent
We verified numerically that when the sample size $n$ and bootstrap replicate $B$ are sufficiently large, the critical quantile $c(1-\alpha)$ obtained from Corollary \ref{lemma:simultaneous confidence bands_bootstrap} converges to the one obtained from Corollary \ref{lemma:simultaneous confidence bands}.
\label{appendix:simultaneousconfidencebands}
\begin{proof}
By Theorem \ref{thm:estimation2}, when the assumptions hold, \begin{align*}
    \sqrt{n}\big\{\widehat \phi_{a, s}(\widetilde{x})-\phi_{a, s}(\widetilde{x})\big\}\rightarrow\mathcal{N}\Big[0, \E_{p_{0}}\big\{\mathit\Phi_ {p_{0}}(a, s, \widetilde{x})^{}\big\}\Big], n\rightarrow\infty.
\end{align*}
Therefore, 
$$
\frac{\widehat \phi_{a, s}(\widetilde{x})-\phi_{a, s}(\widetilde{x})}{\sqrt{\frac{1}{n}\widehat\E_{p_{0}}\{\mathit\Phi_ {p_{0}}(a, s, \widetilde{x})^{2}}\}}=\frac{\widehat \phi_{a, s}(\widetilde{x})-\phi_{a, s}(\widetilde{x})}{\widehat \sigma_{a,s}(\widetilde{x})},
$$ which can be simulated by the empirical Gaussian process $\mathcal{N}_{1}^{b}$ \citep{belloni2015some}.\\
The simultaneous confidence band with coverage probability $(1-\alpha)$ satisfies the following condition:
\begin{align*}
    &P\big\{\widehat \phi_{a, s}(\widetilde{x})-c(1-\alpha)\widehat \sigma_{a,s}(\widetilde{x}) \leqslant \phi_{a, s}(\widetilde{x}) \leqslant \widehat \phi_{a, s}(\widetilde{x})+c(1-\alpha)\widehat \sigma_{a,s}(\widetilde{x}), \text{ }\forall \widetilde{x}\in\widetilde{\mathcal{X}}\big\}=1-\alpha\\
    \implies &P\big\{\text{sup}_{\widetilde{x}\in\widetilde{\mathcal{X}}}\Big|\frac{\widehat \phi_{a, s}(\widetilde{x})-\phi_{a, s}(\widetilde{x})}{\widehat \sigma_{a,s}(\widetilde{x})}\Big|\leqslant c(1-\alpha)\big\}=1-\alpha\\
    \implies &P\big\{\text{max}(|\mathcal{N}_{d}^{b}|)\leqslant c(1-\alpha)\big\}=1-\alpha.
\end{align*}
Therefore, when setting $c(1-\alpha)$ as the $(1-\alpha)$-quantile of $\text{max}(|\mathcal{N}_{d}^{b}|)$, the $(1-\alpha)$ simultaneous confidence bands of $\widehat \phi_{a, s}(\widetilde{x})$ is $\{\widehat \phi_{a, s}(\widetilde{x})-c(1-\alpha)\widehat \sigma_{a,s}(\widetilde{x}),\text{ } \widehat \phi_{a, s}(\widetilde{x})+c(1-\alpha)\widehat \sigma_{a,s}(\widetilde{x})\}$.\\
Heuristically, when the sample size and bootstrap replicates are sufficiently large, the empirical distribution of $\dfrac{\widehat \phi_{a, s}^b(\widetilde{x})-\widehat \phi_{a, s}(\widetilde{x})}{\widehat\sigma_{a,s}(\widetilde{x})}$ in Corollary \ref{lemma:simultaneous confidence bands_bootstrap} approaches to a standard normal distribution, which also justifies the claim of Corollary \ref{lemma:simultaneous confidence bands}.
\end{proof}\clearpage
\section{Identification of $\psi_a(\widetilde{x})$}
\label{appendix:identification1}
\begin{restatable}[Identification of potential outcome means in an external data]{theorem}{thmidentificationcollection}
\label{thm:identification1}
Under conditions \textit{A1} through \textit{A6}, the subgroup potential outcome mean in external data under treatment $a \in \mathcal A$, $\psi_a(\widetilde{x})=\E(Y^a | \widetilde{X}=\widetilde{x}, R = 0 )$, is identifiable by the observed data 
\begin{equation}
\label{eq:identification1g}
  \psi_a(\widetilde{x}) \equiv 
  \E\big\{\E(Y | A=a, X, R=1) | \widetilde{X}=\widetilde{x}, R=0\big\},
\end{equation}
which can be equivalently expressed as 
\begin{equation*} 
  \psi (a,\widetilde{x}) = \dfrac{1}{\Pr(\widetilde{X}=\widetilde{x}, R=0)} \E\Big\{\frac{I(A=a, \widetilde{X}=\widetilde{x}, R=1) Y \Pr(R=0|X)}{\Pr(R=1|X) \Pr(A=a|X, R=1)} \Big\}.
\end{equation*}
\end{restatable}
\begin{proof}
By the law of total expectation and Assumptions \textit{A6, A4, A2, A1} respectively, we have
\begin{align*}
    \psi_a(\widetilde{x})= 
     & \E(Y^{a} | \widetilde{X}=\widetilde{x}, R=0)\\ 
    =& \E\big\{\E(Y^{a} | X, R=0) | \widetilde{X}=\widetilde{x}, R=0\big\}\\ 
    =& \E\big\{\E(Y^{a} | X, R=1) | \widetilde{X}=\widetilde{x}, R=0\big\}\\  
    =& \E\big\{\E(Y^{a} | A, X, R=1) | \widetilde{X}=\widetilde{x}, R=0\big\}\\ 
    =& \E\big\{\E(Y | A=a, X, R=1) | \widetilde{X}=\widetilde{x}, R=0\big\},
\end{align*}
which completes the derivation of the result in \eqref{eq:identification1g}. We can re-express \eqref{eq:identification1g} to use weighting (which relays on Assumptions \textit{A3} and \textit{A5}),
\begin{align*}
    \psi_a(\widetilde{x})
    =& \E\big\{\E(Y | A=a, X, R=1) | \widetilde{X}=\widetilde{x}, R=0\big\}\\
    =& \E\Bigg[\E\Big\{\dfrac{I(A=a, R=1 )Y}{\Pr(R=1|X)\Pr(A=a|X, R=1)}\Big| X\Big\}\Big|\widetilde{X}=\widetilde{x}, R=0\Bigg]\\
    =& \dfrac{1}{\Pr(\widetilde{X}=\widetilde{x}, R=0)}\E\Bigg[I(\widetilde{X}=\widetilde{x}, R=0)E\Big\{\dfrac{I( A=a, R=1)Y}{\Pr(R=1|X)\Pr(A=a|X, R=1)}\Big| X\Big\}\Bigg]\\
    =& \dfrac{1}{\Pr(\widetilde{X}=\widetilde{x}, R=0)}\E\Big\{\dfrac{I(A=a, \widetilde{X}=\widetilde{x}, R=1)Y \Pr(R=0|X)}{\Pr(R=1|X)\Pr(A=a|X, R=1)}\Big\}.
\end{align*}
\end{proof}

\section{The first-order influence function of $\psi_a(\widetilde{x})$ under the non-parametric model}
\label{appendix:influence_function1}
The first-order influence function of $\psi_a(\widetilde{x})$ under the non-parametric model is 
\begin{align*}
\mathit\Psi_ {p_{0}}(a, \widetilde{x})=&\dfrac{1}{\Pr_ {p_{0}}(\widetilde{X}=\widetilde{x}, R=0)}
    \Big[I(\widetilde{X}=\widetilde{x}, R=0)\big\{\E_{p_{0}}(Y | A=a, X, R=1)-\psi_a(\widetilde{x})\big\}\\
    +&\dfrac{\Pr_{p_{0}}(R=0|X)}{1-\Pr_{p_{0}}(R=0|X)}\dfrac{I(A=a, \widetilde{X}=\widetilde{x}, R=1)}{\Pr_{p_{0}}(A=a|X, R=1)}\big\{Y-\E_{p_{0}}(Y | A=a, X, R=1)\big\}\Big].
\end{align*}
\begin{proof}
Recall that under the nonparametric model, $\mathcal M_{\text{\tiny np}}$, for the observable data $O$, the density of the law of the observable data can be written as 
\begin{equation*}
    p(r,x,a,y) = p(r) p(x|r) p(a|x, r) p(y|a, x, r).
\end{equation*}
Under this model, the tangent space is the Hilbert space of mean zero random variables with finite variance. It can be decomposed as $L_2^0 = \Lambda_{R} \oplus \Lambda_{X|R} \oplus  \Lambda_{A|R,X} \oplus  \Lambda_{Y|R, X, A} $. We now derive the influence function for $\psi_a(\widetilde{x})$ under this nonparametric model. 

\noindent
Recall that
\begin{equation*} 
  \psi_a(\widetilde{x}) \equiv \E\big\{\E(Y | A=a, X, R=1) | \widetilde{X}=\widetilde{x}, R=0\big\}.
\end{equation*}
Similar to \ref{appendix:influence_function2}, we have
\begin{align*}
     &\dfrac{\partial \psi_{p_{t}}(a, \widetilde{x})}{\partial t}\Big|_{t=0}\\
    =&\dfrac{\partial}{\partial t}\E_{p_{t}}\big\{\E_{p_{t}}(Y | A=a, X, R=1) | \widetilde{X}=\widetilde{x}, R=0\big\}\Big|_{t=0}\\
    =&\dfrac{\partial}{\partial t}\E_{p_{t}}\big\{\E_{p_{0}}(Y | A=a, X, R=1) | \widetilde{X}=\widetilde{x}, R=0\big\}\Big|_{t=0}\\
    +&\E_{p_{0}}\Big\{\dfrac{\partial}{\partial t}\E_{p_{t}}(Y | A=a, X, R=1)\Big|_{t=0} \widetilde{X}=\widetilde{x}, R=0\Big\}\\
    =&\E_{p_{0}}\big\{E_{p_{0}}(Y | A=a, X, R=1) u_{Y|A, X, R=1} | \widetilde{X}=\widetilde{x}, R=0\big\}\\
    +& \E_{p_{0}}\big\{\E_{p_{0}}(Y u_{Y|A=a, X, R=1}| A=a, X, R=1) |\widetilde{X}=\widetilde{x}, R=0\big\}\\
    =&\E_{p_{0}}\Big\{\dfrac{I(\widetilde{X}=\widetilde{x}, R=0)}{\Pr_ {p_{0}}(\widetilde{X}=\widetilde{x}, R=0)}\E_{p_{0}}(Y | A=a, X, R=1) u_{Y|A, X, R}  \Big\}\\
    +& \E_{p_{0}}\Big(\E_{p_{0}}\Big[\big\{Y-\E_{p_{0}}(Y | A=a, X, R=1)\big\}u_{Y|A=a, X, R=1}| A=a, X, R=1\Big]\Big|\widetilde{X}=\widetilde{x}, R=0\Big)\\
    =&\E_{p_{0}}\Big[\dfrac{I(\widetilde{X}=\widetilde{x}, R=0)}{\Pr_ {p_{0}}(\widetilde{X}=\widetilde{x}, R=0)}\big\{\E_{p_{0}}(Y | A=a, X, R=1)-\psi_a(\widetilde{x})\big\} u_{Y|A, X, R}  \Big]\\
    +& \E_{p_{0}}\Big(\dfrac{I(\widetilde{X}=\widetilde{x}, R=0)}{\Pr_ {p_{0}}(\widetilde{X}=\widetilde{x}, R=0)}\E_{p_{0}}\Big[\big\{Y-\E_{p_{0}}(Y | A=a, X, R=1)\big\}u(O)| A=a, X, R=1\Big]\Big)\\
    =&\E_{p_{0}}\Big[\dfrac{I(\widetilde{X}=\widetilde{x}, R=0)}{\Pr_ {p_{0}}(\widetilde{X}=\widetilde{x}, R=0)}\big\{\E_{p_{0}}(Y | A=a, X, R=1)-\psi_a(\widetilde{x})\big\} u(O)  \Big]\\
    +& \E_{p_{0}}\Big[\E_{p_{0}}\Big\{ \dfrac{I(\widetilde{X}=\widetilde{x}, R=0)}{\Pr_ {p_{0}}(\widetilde{X}=\widetilde{x}, R=0)}\Big| A=a, X, R=1\Big\}
    \big\{Y-\E_{p_{0}}(Y | A=a, X, R=1)\big\}u(O)\Big]\\
    =&\E_{p_{0}}\Big(u(O)\dfrac{I(\widetilde{X}=\widetilde{x})}{\Pr_ {p_{0}}(\widetilde{X}=\widetilde{x}, R=0)}
    \Big[I(R=0)\big\{\E_{p_{0}}(Y | A=a, X, R=1)-\psi_a(\widetilde{x})\big\}\\
    +&R\dfrac{\Pr_{p_{0}}(R=0|X)}{1-\Pr_{p_{0}}(R=0|X)}\dfrac{I(A=a)}{\Pr_{p_{0}}(A=a|X, R=1)}\big\{Y-\E_{p_{0}}(Y | A=a, X, R=1)\big\}\Big]
    \Big).
\end{align*}
Thus, we conclude that the influence function of $\psi_a(\widetilde{x})$ is $\mathit\Psi_{p_0}(a, \widetilde{x})$.
\end{proof}
\section{}
\label{app:cor:sameIF_phi}
\begin{restatable}[Efficient influence function under semiparametric model]{corollary}{cor:identification1}
\label{cor:identification1}
The efficient influence function under the semi-parametric model which incorporates the propensity score $\Pr(A=a|X, R=1)$ (e.g., the multi-source data are a collection of randomized clinical trials) is $\mathit\Psi_ {p_{0}}(a, \widetilde{x})$.
\end{restatable}
The proof is similar to the proof of Corollary \ref{cor:identification2}, and thus we omit it here.
\section{Proof of Theorem~\ref{thm:estimation1}}
\label{app:Proof-Estimation1}
\begin{restatable}{theorem}{thmestimationSA}
\label{thm:estimation1}
If assumptions \textit{A1} through \textit{A6}, and \textit{(b1)} through \textit{(b3)} hold, then 
\begin{align*}
     \widehat \psi_a(\widetilde{x})-\psi_a(\widetilde{x})
    =\mathbb{P}_n\{\mathit\Psi_ {p_{0}}(a, \widetilde{x})\}+
    Q_{n}+o_p(n^{-1/2}),
\end{align*}
where $Q_{n}\lesssim O_p\big\{||\widehat g_a(X) -g_a(X)||(||\widehat e_a(X) -e_a(X)||+||\widehat p_a(X) -p_a(X)||)\big\}$, and nuisance parameters are estimated on a separate independent sample.\\
In particular, if $Q_{n} = o_p(n^{-1/2})$, then 
$
    \sqrt{n}\big\{\widehat \psi_a(\widetilde{x})-\psi_a(\widetilde{x})\big\}\rightarrow\mathcal{N}\Big[0, \E_{p_{0}}\big\{\mathit\Psi_ {p_{0}}(a, \widetilde{x})^{2}\big\}\Big].
$\\
That is, $\widehat \psi_a(\widetilde{x})$ is non-parametric efficient.
\end{restatable}
\begin{proof}
Define all nuisance parameters (asymptotic limits) as $\boldsymbol{\theta}\equiv \big\{\gamma(\widetilde{x}), g_{a}(X), e_a(X)$, $p(X)\big\}$, and estimates as $\boldsymbol{\widehat \theta}\equiv \big\{\widehat \gamma(\widetilde{x}), \widehat g_{a}(X), \widehat e_a(X)$, $\widehat p(X)\big\}$. 
For general functions $\gamma'(\widetilde{x}), g'_{a}(X), e'_{a} (X),$ and $p'(X)$ define 
\begin{align*}
H(\boldsymbol{\theta}') =   \gamma'(\widetilde{x}) \Bigg[ I(\widetilde{X}=\widetilde{x}, R = 0) g'_{a}(X) +
&I(A =a, \widetilde{X}=\widetilde{x}, R=1 ) \dfrac{  1 - p' (X) }{p'(X) e'_{a} (X)} \Big\{Y  -  g'_{a}(X)  \Big\}    \Bigg].
\end{align*}
We observe that we can decompose $\widehat \psi_a(\widetilde{x}) - \psi_a(\widetilde{x})$ into three parts as below.
\begin{align*}
     &\widehat \psi_a(\widetilde{x}) - \psi_a(\widetilde{x})=\mathbb{P}_n\big\{H(\boldsymbol{\widehat \theta})\big\}-\mathbb{P}\big\{H(\boldsymbol{\theta})\big\}\\
    =&\underbrace{
    (\mathbb{P}_n-\mathbb{P})\big\{H(\boldsymbol{\widehat \theta})-H(\boldsymbol{\theta})\big\} }_{1}+
    \underbrace{
    \mathbb{P}\big\{H(\boldsymbol{\widehat \theta})-H(\boldsymbol{\theta})\big\} }_{2}+
    \underbrace{
    (\mathbb{P}_n-\mathbb{P})H(\boldsymbol{\theta})}_{3}.
\end{align*}
Working on term 1, note that
\begin{align*}
    &||H(\boldsymbol{\widehat \theta})-H(\boldsymbol{\theta})||\\
    = 
    & \Bigg|\Bigg|I(\widetilde{X}=\widetilde{x})\widehat \gamma(\widetilde{x})\Bigg[(1-R)\widehat g_{a}(X)+R\dfrac{I(A=a)}{\widehat e_a(X) }\dfrac{1-\widehat p_a(X)}{\widehat p_a(X)}\Big\{Y-\widehat \mu_a() \Big\}  \Bigg]\\
    & -I(\widetilde{X}=\widetilde{x}) \gamma(\widetilde{x})\Bigg[(1-R) g_{a}(X)+R\dfrac{I(A=a)}{ e_a(X) }\dfrac{1- p_a(X)}{ p_a(X)}\Big\{Y- g_a(X) \Big\}  \Bigg]\Bigg|\Bigg|\\
    \leqslant 
    & \Bigg|\Bigg|I(\widetilde{X}=\widetilde{x})\widehat \gamma(\widetilde{x})\Bigg[(1-R)\{\widehat g_a(X)- g_a(X)\}+R\dfrac{I(A=a)}{\widehat e_a(X) }\dfrac{1-\widehat p_a(X)}{\widehat p_a(X)}\Big\{ g_a(X)- \widehat g_a(X) \Big\}  \Bigg]\Bigg|\Bigg|\\
    & + \Bigg|\Bigg|I(\widetilde{X}=\widetilde{x}) \{\widehat \gamma(\widetilde{x})-\gamma(\widetilde{x})\}\Bigg[(1-R) g_{a}(X)+R\dfrac{I(A=a)}{ e_a(X) }\dfrac{1- p_a(X)}{ p_a(X)}\Big\{Y- g_a(X) \Big\}  \Bigg]\Bigg|\Bigg|\\
    & + \Bigg|\Bigg|I(\widetilde{X}=\widetilde{x})\widehat \gamma(\widetilde{x})RI(A=a) \Big\{\dfrac{1-\widehat p_a(X)}{\widehat e_a(X) \widehat p_a(X)}-\dfrac{1- p_a(X)}{ e_a(X)  p_a(X)} \Big\}\Big\{ Y- g_a(X)\Big\}\Bigg|\Bigg|\\
    \lesssim
    & || \widehat g_a(x)-g_a(X)||+||\widehat e_a(x)-e_a(X)||+||\widehat p_a(x)-p_a(X)|| +||\widehat \gamma(\widetilde{x})-\gamma(\widetilde{x})||
\end{align*}
By similar argument to Lemma \ref{lemma:gamma} and assumptions \textit{(b3)}, we have $||H(\boldsymbol{\widehat \theta})-H(\boldsymbol{\theta})|| = o_p(1)$, so that by Lemma 2 of \citet{kennedy2020sharp}, $(\mathbb{P}_n-\mathbb{P})\big\{H(\boldsymbol{\widehat \theta})-H(\boldsymbol{\theta})\big\}=o_p(n^{-1/2})$.

Working on term 2, we have
\begin{align*}
     &\E\{H(\boldsymbol{\widehat \theta})-H(\boldsymbol{\theta})\}\\
    =& \E\Bigg(I(\widetilde{X}=\widetilde{x})
    \Big\{
    \widehat \gamma(\widetilde{x})\bigg[(1-R)\widehat g_a(X)+R\dfrac{I(A=a)}{\widehat e_a(X) }\dfrac{1-\widehat p_a(X)}{\widehat p_a(X)}\Big\{ Y- \widehat g_a(X) \Big\} 
    \bigg]\\
    &-\gamma(\widetilde{x})(1-R) g_a(X)
    \Big\}\Bigg) \\
    =&\E\Bigg\{I(\widetilde{X}=\widetilde{x})
    \Big(\widehat \gamma(\widetilde{x})\Big[\{1- p_a(X)\}\{\widehat g_a(X) - g_a(X) \}\\
    & +\dfrac{1-\widehat p_a(X)}{\widehat p_a(X)} p_a(X)\dfrac{e_a(x)}{\widehat e_a(X)}\Big\{\widehat g_a(X) - g_a(X) \Big\}\Big]
    \Big)\Bigg\}\\
    & + \{\widehat \gamma(\widetilde{x})- \gamma(\widetilde{x})\}\E\big[\{1-p_a(X)\}g_a(X)\big]\\
    =& \E\Bigg\{I(\widetilde{X}=\widetilde{x})
    \Big(\widehat \gamma(\widetilde{x})\Big[\{\widehat p_a(X)- p_a(X)\}\{\widehat g_a(X) - g_a(X) \}\\
    & +\{1-\widehat p_a(X)\} \Big\{1-\dfrac{p_a(X)}{\widehat p_a(X)}\dfrac{e_a(X)}{\widehat e_a(X)}\Big\}\Big\{\widehat g_a(X) - g_a(X) \Big\}\Big]
    \Big)\Bigg\}+ \Big\{\dfrac{\widehat\gamma(\widetilde{x})}{\gamma(\widetilde{x})} -1\Big\}\psi_a(\widetilde{x}).
\end{align*}
Combine term 2 and 3, we have
\begin{align}\nonumber
     &\mathbb{P}\big\{H(\boldsymbol{\widehat \theta})-H(\boldsymbol{\theta})\big\}+
    (\mathbb{P}_n-\mathbb{P})H(\boldsymbol{\theta})\\\label{1}
    =& \mathbb{P}_n\{\mathit\Psi_{p_0}(a, \widetilde{x})\}\\\nonumber
    &+\E\Bigg\{I(\widetilde{X}=\widetilde{x})
    \Big(\widehat \gamma(\widetilde{x})\Big[\{\widehat p_a(X)- p_a(X)\}\{\widehat g_a(X) - g_a(X) \}\\\label{2}
    & +\{1-\widehat p_a(X)\} \Big\{1-\dfrac{p_a(X)}{\widehat p_a(X)}\dfrac{e_a(X)}{\widehat e_a(X)}\Big\}\Big\{\widehat g_a(X) - g_a(X) \Big\}\Big]
    \Big)\Bigg\}\\
    &+\Big\{\dfrac{\widehat\gamma(\widetilde{x})^{-1}}{\gamma(\widetilde{x})^{-1}} -1\Big\}\psi_a(\widetilde{x})+\Big\{\dfrac{\widehat\gamma(\widetilde{x})}{\gamma(\widetilde{x})} -1\Big\}\psi_a(\widetilde{x})\label{3}
\end{align}
Factoring and simplifying term (~\ref{3}) yields $\psi_a(\widetilde{x})\left\{\widehat\gamma(\widetilde{x}) - \gamma(\widetilde{x})\right\}\left\{\frac{1}{\gamma(\widetilde{x})} - \frac{1}{\widehat\gamma(\widetilde{x})}\right\}$, which is\\ 
$o_{p}(n^{-1/2})$ by the central limit theorem and similar argument to Lemma \ref{lemma:gamma}. That is, both term 1 and (\ref{3}) are $o_{p}(n^{-1/2})$.

Therefore, combining terms 1, 2 and 3, we conclude
\begin{align*}
     \widehat \psi_a(\widetilde{x})-\psi_a(\widetilde{x})
    =&
    O_{p}\Big[||\widehat g_a(X) -g_a(X)||\big\{||\widehat e_a(X) -e_a(X)||+||\widehat p_a(X) -p_a(X)||\big\}\Big]+\\
    &\mathbb{P}_n\{\mathit\Psi_ {p_{0}}(a, \widetilde{x})\}+o_p(n^{-1/2}).
\end{align*}
Further, if $\left\lVert \widehat{g}_a - g_a\right\rVert \left\{\left\lVert \widehat{e}_a - e_a\right\rVert + \left\lVert \widehat{p}_a - p_a\right\rVert\right\} = o_p(n^{-1/2})$, then
\begin{align*}
    \sqrt{n}\big\{\widehat \psi_a(\widetilde{x})-\psi_a(\widetilde{x})\big\}\rightarrow\mathcal{N}\Big[0, \E_{p_{0}}\big\{\mathit\Psi_ {p_{0}}(a, \widetilde{x})^{2}\big\}\Big].
\end{align*}
That is, $\widehat \psi_a(\widetilde{x})$ is non-parametric efficient.
\end{proof}
\section{Data generation mechanism of simulation for rate robustness assessment}
\label{app:sim_gen_2}
For simplicity, we considered a two data source setting with $S \in \{1,2\}$. First, we generated $X = (X_0, X_1)$ with $X_0 \sim \mathrm{Bernoulli}(0.5)$ independent of $X_1 \sim \mathrm{Uniform}(0,1)$, then sampled $S \mid X$ with $\Pr(S = 1 \mid X) = \mathrm{expit}(-0.2 + 0.5 X_0 + 1.2 X_1)$. Next, we generated $A \mid X, S \sim \min{\{\max{[0.1, \Pr_A(X, S)]}, 0.9\}}$, where
$\Pr_A(X, S) = \mathrm{expit}(0.8 + 0.9X_0 - 0.8X_1)  \text{ if } S = 1, \text{ and } \Pr_A(X, S) = \mathrm{expit}(-0.8 -0.9X_0 + 0.8 X_1)  \text{ if } S = 0$.
 Finally, we generated $Y \mid X, S, A \sim \mathcal{N}(m_Y(X, S, A), 1)$, where
$m_Y(X, S, A) = 5.2 + (1.2 - 0.6 X_0)A + X_0 - 1.2X_1$.
\section{Additional simulation results}\label{app:add_sim}
Figure \ref{app:fig:r1} show the biases, coverage of simultaneous confidence bands, coverage of pointwise confidence intervals,theoretical standard deviation and Monte-Carlo standard deviation of $\widehat \psi_{1}(\widetilde{x}), \widetilde{x}=1, 2, 3, 4, 5$. Sample size is 1000 with 5000 iterations.\\
\begin{figure}[h]
    \centering
    \includegraphics[width=0.75\textwidth]{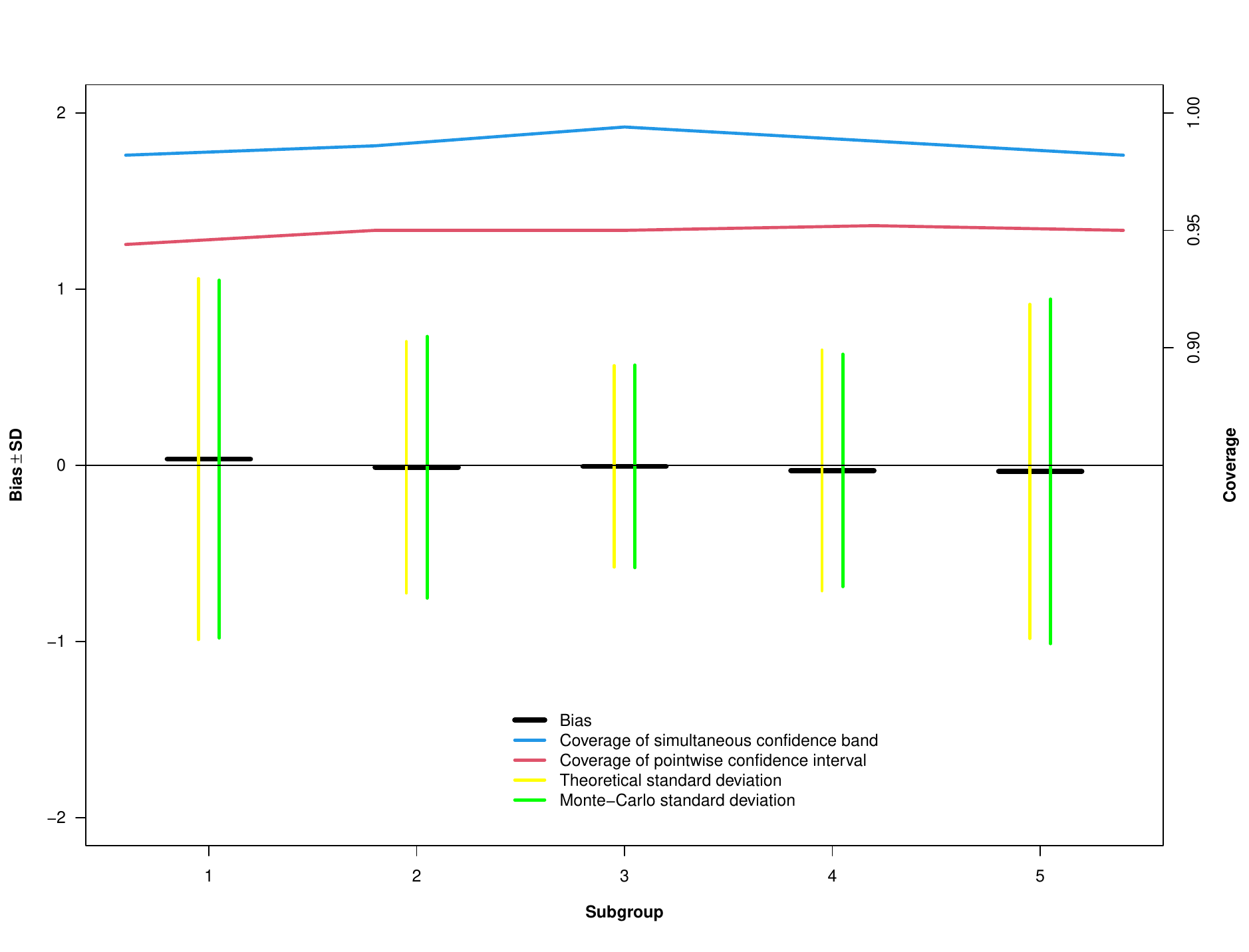}
    \caption{Biases, coverage of simultaneous confidence bands, coverage of pointwise confidence intervals,theoretical standard deviation and Monte-Carlo standard deviation of $\widehat \psi_{1}(\widetilde{x}), \widetilde{x}=1, 2, 3, 4, 5$. Sample size is 1000 with 5000 iterations.}
    \label{app:fig:r1}
\end{figure}
Figure \ref{app:fig:r2} shows the simulation results for $\widehat \psi_{1}(3)$ when the sample size $n=\sum_{s=0}^{3}n_s-10^4$ and $\sum_{s=1}^{3}n_s=1000$ with 5000 iterations. The results include the averaged biases and Monte-Carlo standard errors of various estimators with different combination(s) of correctly specified models. The methods being compared are the proposed doubly robust estimator (DR) $\widehat\psi_{1}(3)$, the estimator that only uses outcome regression (Regression) $\widehat\psi^g_{1}(3)$, and the inverse probability of treatment weighting (IPTW) $\widehat\psi^w_{1}(3)$.\\
\begin{figure}
    \centering
    \includegraphics[width=0.75\textwidth,trim={0 0 0 2cm},clip]{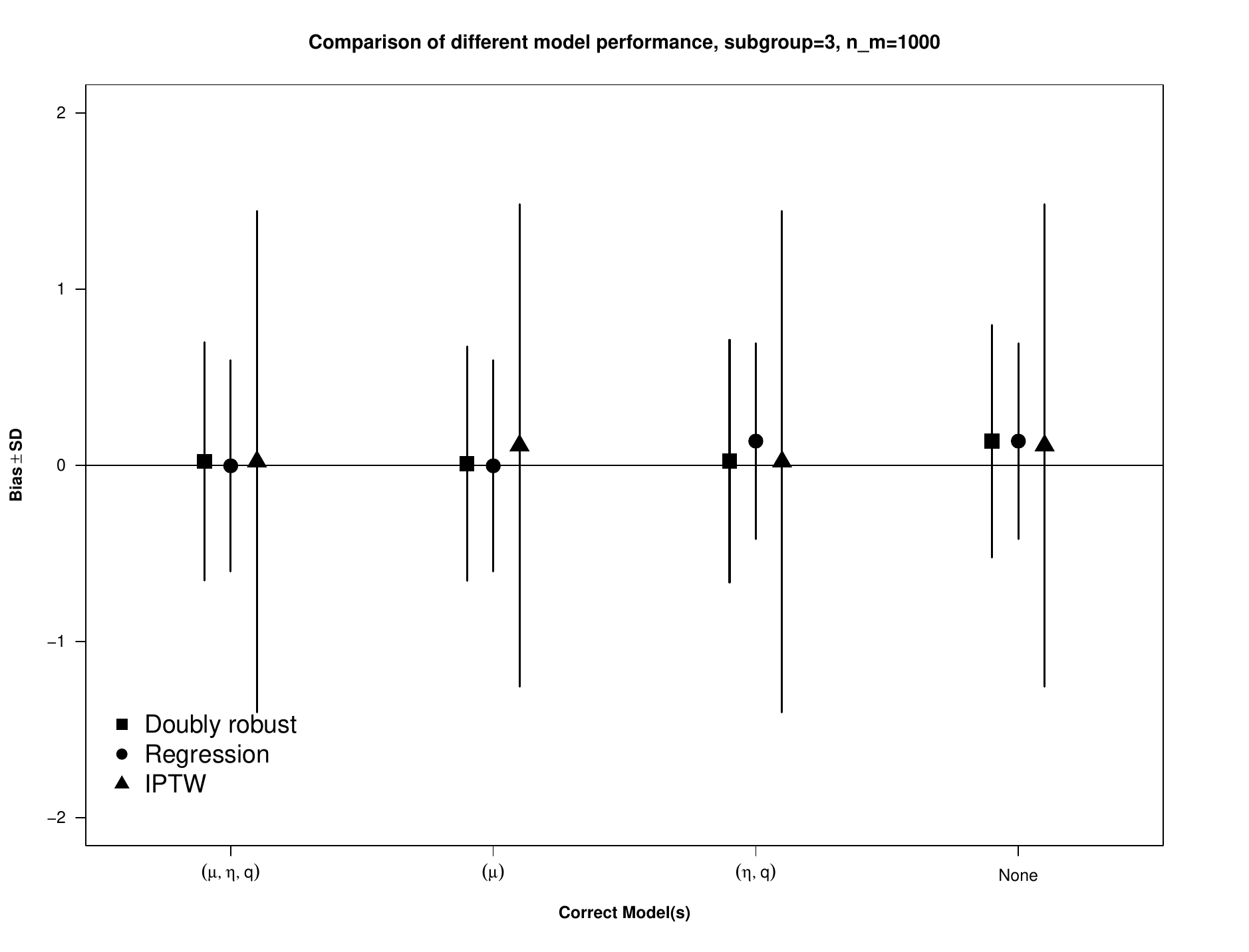}
    \caption{Simulation results for $\widehat \psi_{1}(3)$ when the sample size $n=\sum_{s=0}^{3}n_s-10^4$ and $\sum_{s=1}^{3}n_s=1000$ with 5000 iterations. The results include the averaged biases and Monte-Carlo standard errors of various estimators with different combination(s) of correctly specified models. The methods being compared are the proposed doubly robust estimator (DR) $\widehat\psi_{1}(3)$, the estimator that only uses outcome regression (Regression) $\widehat\psi^g_{1}(3)$, and the inverse probability of treatment weighting (IPTW) $\widehat\psi^w_{1}(3)$.}
    \label{app:fig:r2}
\end{figure}
Table \ref{app:table:add_s} shows the simulation results when the sample sizes are 1000, 2000, 5000. The results include the averaged biases and Monte-Carlo standard errors of various estimators with different combination(s) of correctly specified models from 500 simulations. The methods being compared are the proposed doubly robust estimator (DR) $\widehat\phi_{1,1}(3)$, the estimator that only uses outcome regression (Regression) $\widehat\phi^g_{1,1}(3)$, and the inverse probability of treatment weighting (IPTW) $\widehat\phi^w_{1,1}(3)$.\\
\begin{table}[h]
\centering
\begin{tabular}{cccccccc}
  \hline
  &&\multicolumn{2}{c}{\textbf{$\sum_{s=1}^{3}n_s=1000$}}&\multicolumn{2}{c}{$\sum_{s=1}^{3}n_s=2000$}&\multicolumn{2}{c}{$\sum_{s=1}^{3}n_s=5000$}\\
    \cline{3-4}\cline{5-6}\cline{7-8}
 \textbf{Correct models}& \textbf{Method} & \textbf{Bias} & \textbf{SD}& \textbf{Bias} & \textbf{SD}& \textbf{Bias} & \textbf{SD} \\ 
  \hline
\multirow{3}{*}{$(\mu,q,\eta)$}
&DR & 0.00 & 0.62 & 0.00 & 0.44 &0.00 & 0.28 \\ 
&Regression & 0.01 & 0.58 & 0.00 & 0.43 & 0.00 & 0.26 \\ 
&IPTW & 0.03 & 0.97 & 0.02 & 0.65 & 0.00 & 0.44 \\ \hline
\multirow{3}{*}{$(\mu)$}  
&DR & 0.01 & 0.62 & 0.00 & 0.44 & 0.00 & 0.28 \\ 
  &Regression & 0.01 & 0.58 & 0.00 & 0.43 & 0.00 & 0.26 \\ 
  &IPTW & 0.14 & 0.97 & 0.11 & 0.66 & 0.09 & 0.45 \\ \hline
\multirow{3}{*}{$(q,\eta)$}  
&DR & 0.00 & 0.62 & 0.00 & 0.44 & 0.00 & 0.28 \\ 
  &Regression & 0.14 & 0.55 & 0.12 & 0.42 & 0.11 & 0.26 \\ 
  &IPTW & 0.03 & 0.97 & 0.02 & 0.65 & 0.00 & 0.44 \\ \hline
\multirow{3}{*}{None}  
&DR & 0.11 & 0.62 & 0.10 & 0.43 & 0.09 & 0.28 \\ 
  &Regression & 0.14 & 0.55 & 0.12 & 0.42 & 0.11 & 0.26 \\ 
  &IPTW & 0.14 & 0.97 & 0.11 & 0.66 & 0.09 & 0.45 \\ 
   \hline
\end{tabular}
\caption{Simulation results when the sample sizes are 1000, 2000, 5000. The results include the averaged biases and Monte-Carlo standard errors of various estimators with different combination(s) of correctly specified models from 500 simulations. The methods being compared are the proposed doubly robust estimator (DR) $\widehat\phi_{1,1}(3)$, the estimator that only uses outcome regression (Regression) $\widehat\phi^g_{1,1}(3)$, and the inverse probability of treatment weighting (IPTW) $\widehat\phi^w_{1,1}(3)$.}
\label{app:table:add_s}
\end{table}
Table \ref{app:table:add_r1} shows the simulation results when the sample sizes $\sum_{s=1}^{3}n_s$ is 1000, 2000, 5000 and $n=\sum_{s=0}^{3}n_s$ is $10^4$. The results include the averaged biases and Monte-Carlo standard errors of various estimators with different combination(s) of correctly specified models from 500 simulations. The methods being compared are the proposed doubly robust estimator (DR) $\widehat\psi_{1}(3)$, the estimator that only uses outcome regression (Regression), and the inverse probability of treatment weighting (IPTW).\\
\begin{table}[h]
\centering
\begin{tabular}{cccccccc}
  \hline
    &&\multicolumn{2}{c}{\textbf{$\sum_{s=1}^{3}n_s=1000$}}&\multicolumn{2}{c}{$\sum_{s=1}^{3}n_s=2000$}&\multicolumn{2}{c}{$\sum_{s=1}^{3}n_s=5000$}\\
    \cline{3-4}\cline{5-6}\cline{7-8}
 \textbf{Correct models}& \textbf{Method} & \textbf{Bias} & \textbf{SD}& \textbf{Bias} & \textbf{SD}& \textbf{Bias} & \textbf{SD} \\ 
  \hline
  \multirow{3}{*}{$(g,p,e)$}
&DR &          -0.01 & 0.71 &-0.01 & 0.51 &-0.01 & 0.32\\
&Regression &  0.01  &  0.61 &  -0.02  &  0.45  & -0.01  &  0.27\\
&IPTW &  0.07  &  1.45 &  -0.02  &  1.01 &   0.03  &  0.71\\
 \multirow{3}{*}{$(p,e)$}
&DR &  -0.01  &  0.69  & -0.02  &  0.51  & -0.01  &  0.31\\
&Regression & 0.01  &  0.61  & -0.02  &  0.45  & -0.01  &  0.27\\
&IPTW &  0.13   & 1.38  &  0.10  &  0.99  &  0.16  &  0.70\\
\multirow{3}{*}{$(p,e)$}
&DR  & -0.01  &  0.72  & -0.01  &  0.52  & -0.01  &  0.32\\
&Regression &  0.16  &  0.58  &  0.13  &  0.42  &  0.13  & 0.26\\
&IPTW & 0.07 &   1.45  & -0.02  &  1.01  &  0.03  &  0.71\\
\multirow{3}{*}{None}
&DR  &  0.11   & 0.68  &  0.11  &  0.51  &  0.12  &  0.30\\
&Regression &  0.16   & 0.58  &  0.13  &  0.42  &  0.13   & 0.26\\
&IPTW & 0.13  &  1.38   & 0.10   & 0.99  &  0.16   & 0.70\\
   \hline
\end{tabular}
\caption{Simulation results when the sample sizes $\sum_{s=1}^{3}n_s$ is 1000, 2000, 5000 and $n=\sum_{s=0}^{3}n_s$ is $10^4$. The results include the averaged biases and Monte-Carlo standard errors of various estimators with different combination(s) of correctly specified models from 500 simulations. The methods being compared are the proposed doubly robust estimator (DR) $\widehat\psi_{1}(3)$, the estimator that only uses outcome regression (Regression), and the inverse probability of treatment weighting (IPTW).}
\label{app:table:add_r1}
\end{table}

Table \ref{app:table:add_r2} shows the simulation results when the sample sizes $\sum_{s=1}^{3}n_s$ is 1000, 2000, 5000 and $n=\sum_{s=0}^{3}n_s$ is $10^5$. The results include the averaged biases and Monte-Carlo standard errors of various estimators with different combination(s) of correctly specified models from 500 simulations. The methods being compared are the proposed doubly robust estimator (DR) $\widehat\psi_{1}(3)$, the estimator that only uses outcome regression (Regression), and the inverse probability of treatment weighting (IPTW).
\begin{table}[h]
\centering
\begin{tabular}{cccccccc}
  \hline
    &&\multicolumn{2}{c}{\textbf{$\sum_{s=1}^{3}n_s=1000$}}&\multicolumn{2}{c}{$\sum_{s=1}^{3}n_s=2000$}&\multicolumn{2}{c}{$\sum_{s=1}^{3}n_s=5000$}\\
    \cline{3-4}\cline{5-6}\cline{7-8}
 \textbf{Correct models}& \textbf{Method} & \textbf{Bias} & \textbf{SD}& \textbf{Bias} & \textbf{SD}& \textbf{Bias} & \textbf{SD} \\ 
  \hline
  \multirow{3}{*}{$(g,p,e)$}
  &DR         & -0.01  &  0.71  &  0.01  &  0.52  &  0.00  &  0.33 \\ 
  &Regression & 0.00  &  0.63  &  0.01  &  0.45 &  -0.01  &  0.28 \\ 
  &IPTW       & -0.05  &  1.43   & 0.04  &  0.91  &  0.02  &  0.64 \\ 
  \multirow{3}{*}{$(g)$}
  &DR         & -0.01  &  0.69  &  0.01  &  0.50  &  0.00  &  0.32\\ 
  &Regression & 0.00  &  0.63 &   0.01  &  0.45  & -0.01  &  0.28 \\ 
  &IPTW       & 0.17  &  1.36  &  0.15  &  0.87  &  0.15  &  0.62  \\ 
  \multirow{3}{*}{$(p,e)$}
  &DR         & -0.02  &  0.71  &  0.01  &  0.52  &  0.00  &  0.33\\ 
  &Regression & 0.15  &  0.59  &  0.15  &  0.43  &  0.14   & 0.27 \\ 
  &IPTW       &  -0.05  &  1.43  &  0.04  &  0.91  &  0.02  &  0.64\\ 
  \multirow{3}{*}{None}
  &DR         & 0.11  &  0.68   & 0.13  &  0.50  &  0.13  &  0.31\\ 
  &Regression & 0.15  &  0.59   & 0.15  &  0.43  &  0.14  &  0.27\\ 
  &IPTW       & 0.17  &  1.36  &  0.15  &  0.87  &  0.15  &  0.62\\
   \hline
\end{tabular}
\caption{Simulation results when the sample sizes $\sum_{s=1}^{3}n_s$ is 1000, 2000, 5000 and $n=\sum_{s=0}^{3}n_s$ is $10^5$. The results include the averaged biases and Monte-Carlo standard errors of various estimators with different combination(s) of correctly specified models from 500 simulations. The methods being compared are the proposed doubly robust estimator (DR) $\widehat\psi_{1}(3)$, the estimator that only uses outcome regression (Regression), and the inverse probability of treatment weighting (IPTW).}
\label{app:table:add_r2}
\end{table}

\section{Summary statistics of each trial}
\label{app:application1}
Table \ref{tab:TableOne} shows the baseline characteristics in the meta-trials, stratified by $S$. Results reported as mean (standard deviation) for continuous variables and count (percentage) for binary variables.\\
\begin{table}[]
    \centering
    \begin{tabular}{cccc}
    \hline
    Source of data, S     &  1 & 2 & 3\\
    \hline
    Number of individuals & 485 & 284 & 448\\
    PANSS at 6 week & 78.34 (23.74) & 8013 (21.66) & 79.13 (22.98)\\
    Paliperidone  ER & 0.75 (0.43) & 0.68 (0.47) & 0.75 (0.43)\\
    PANSS Baseline & 93.6 (10.97) & 93.26 (11.57) & 92.98 (12.46)\\
    PANSS Baseline $<$ 85 & 96 (19.8) & 72 (25.4) & 119 (26.6)\\
    PANSS Baseline, 85-100 & 253 (52.1) & 135 (47.5) & 204 (45.5)\\
    PANSS Baseline $>$ 100 & 137 (28.2) & 77 (27.1) & 125 (27.9)\\
    Age & 42.64 (11.42) & 41.87 (10.71) & 37.01 (10.63)\\
    Age Group, 18-34 & 215 (44.2) & 71 (25.0) & 194 (43.3)\\
    Age Group, 34-48 & 119 (24.5) & 136 (47.9) & 182 (40.6)\\
    Age Group, 48-74 & 152 (31.3) & 77 (27.1) & 72 (16.1)\\
    Female & 0.47 (0.50) & 0.29 (0.45) & 0.36 (0.48)\\
    Race, White & 420 (86.4) & 129 (45.4) & 219 (48.9)\\
    Race, Black & 0 (0.0) & 151(53.2) & 90 (20.1)\\
    Race, Asian & 1 (0.2) & 0 (0.0) & 112 (25.0)\\
    Race, Other & 65 (13.4) & 4 (1.4) & 27 (6.0)\\
    \hline
    \end{tabular}
    \caption{Baseline characteristics in the meta-trials, stratified by $S$. Results reported as mean (standard deviation) for continuous variables and count (percentage) for binary variables.}
    \label{tab:TableOne}
\end{table}

\section{Additional results}
\label{app:application2}
Figure \ref{fig:gender} shows the results from analyses using the meta-trial data. For each trail (target population), and each gender subgroup, the plot includes the point estimation (horizontal short solid line), 95\% confidence interval (vertical solid line), and simultaneous confidence bands (vertical dashed line). The horizontal dashed line indicates the null.\\
\begin{figure}
    \centering
    \includegraphics[width=0.5\textwidth]{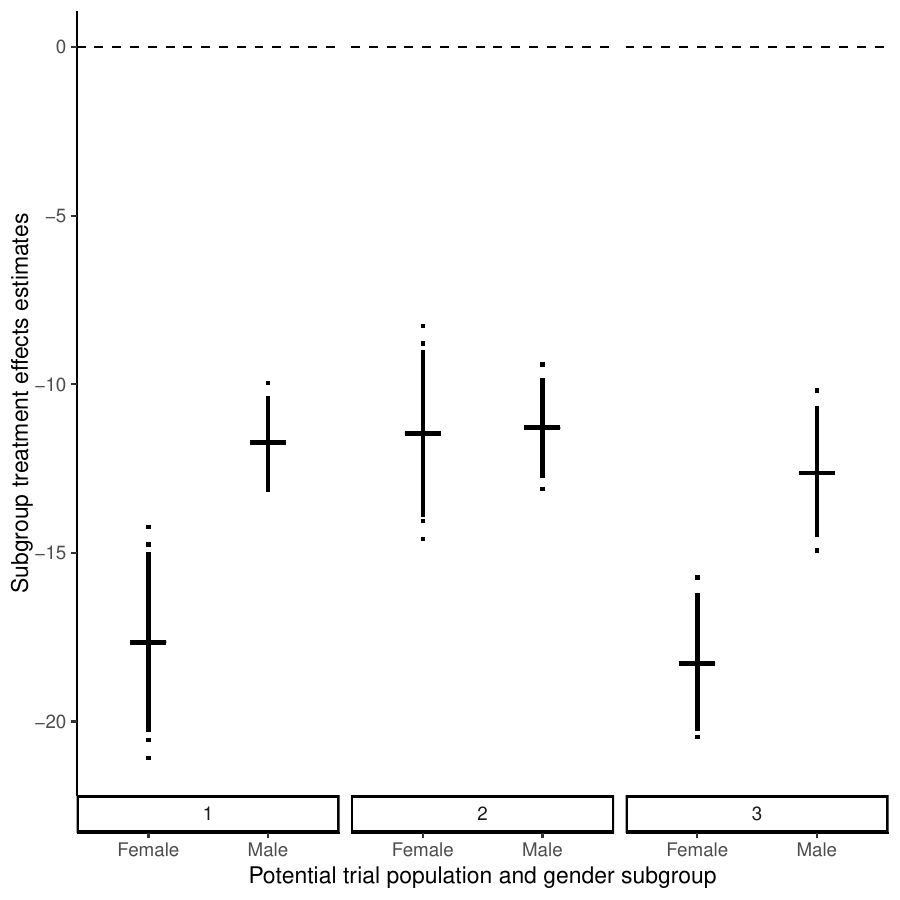}
    \caption{Results from analyses using the meta-trial data. For each trail (target population), and each gender subgroup, the plot includes the point estimation (horizontal short solid line), 95\% confidence interval (vertical solid line), and simultaneous confidence bands (vertical dashed line). The horizontal dashed line indicates the null.}
    \label{fig:gender}
\end{figure}
Figure \ref{fig:agebaseline} shows the results from analyses using the meta-trial data. For each trail (target population), and each PANSS subgroup, the plot includes the point estimation (horizontal short solid line), 95\% confidence interval (vertical solid line), and simultaneous confidence bands (vertical dashed line). The horizontal dashed line indicates the null.
\begin{figure}
    \centering
    \includegraphics[width=0.5\textwidth]{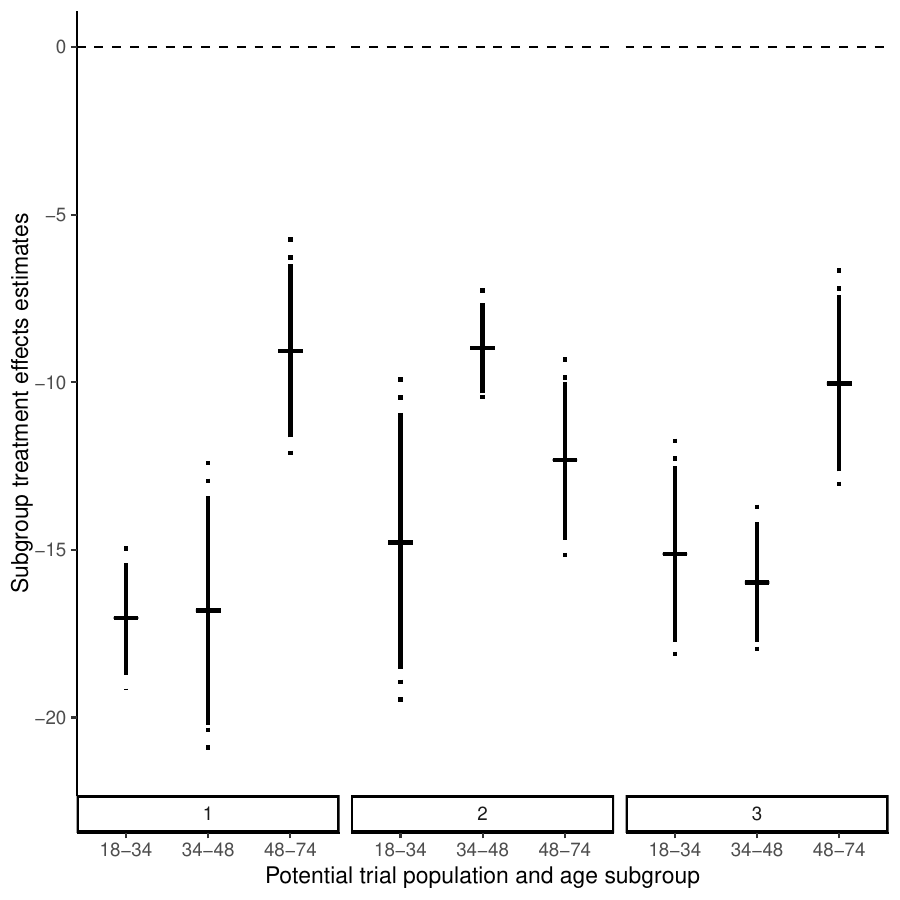}
    \caption{Results from analyses using the meta-trial data. For each trail (target population), and each age subgroup, the plot includes the point estimation (horizontal short solid line), 95\% confidence interval (vertical solid line), and simultaneous confidence bands (vertical dashed line). The horizontal dashed line indicates the null.}
    \label{fig:agebaseline}
\end{figure}
\end{appendices}
\clearpage
\end{document}